\documentclass[authoryear,3p]{elsarticle}
\usepackage[utf8]{inputenc}
\usepackage{amsmath}
\usepackage{amsfonts}
\usepackage{amssymb}
\usepackage{amsfonts}
\usepackage{amsthm}
\usepackage{amscd}
\usepackage{color}
\usepackage{graphics}
\usepackage{graphicx}
\usepackage{psfrag}
\usepackage[toc,page]{appendix}

\usepackage[utf8]{inputenc}
\usepackage{tikz}
\definecolor{ccccccc}{RGB}{204,204,204}
\definecolor{cffffff}{RGB}{255,255,255}
\definecolor{cff0000}{RGB}{255,0,0}
\definecolor{c0000ff}{RGB}{0,0,255}
\definecolor{c00ff00}{RGB}{0,255,0}

\usepackage[all]{xy}

\usepackage{algorithm,algorithmic}

\usepackage{natbib}

\newtheorem{theorem}{\bf Theorem}[section]
\newtheorem{lemma}[theorem]{\bf Lemma}
\newtheorem{proposition}[theorem]{\bf Proposition}

\newtheorem{remark}[theorem]{\bf Remark}

\newtheorem{example}[theorem]{\bf Example}

\newcommand{\f}[1]{\mathbf{#1}}

\newcommand{\hf}{\ensuremath{\mathrm{HF}}}

\newcommand{\C}{\mathbb{C}}
\newcommand{\R}{\mathbb{R}}

\journal{Computer Aided Geometric Design}

\begin{document}

\sloppy

\begin{frontmatter}

\title{Hermite interpolation by piecewise polynomial surfaces with polynomial~area~element}

\author[plzen2]{Michal Bizzarri}
\ead{bizzarri@ntis.zcu.cz}

\author[plzen1,plzen2]{Miroslav L\'{a}vi\v{c}ka}
\ead{lavicka@kma.zcu.cz}

\author[praha]{Zbyn\v{e}k \v{S}\'{i}r}
\ead{zbynek.sir@mff.cuni.cz}

\author[plzen1,plzen2]{Jan Vr\v{s}ek\corref{cor1}}
\cortext[cor1]{Corresponding author}
\ead{vrsekjan@kma.zcu.cz}

\address[plzen1]{Department of Mathematics, Faculty of Applied Sciences, University of West Bohemia,
         Univerzitn\'i~8,~306~14~Plze\v{n},~Czech~Republic}

\address[plzen2]{NTIS -- New Technologies for the Information Society, Faculty of Applied Sciences, University of West Bohemia, Univerzitn\'i 8, 306 14 Plze\v{n}, Czech~Republic}

\address[praha]{Mathematical Institute,  Charles University, Sokolovská 83, 186 75 Praha, Czech~Republic}

\begin{abstract}
This paper is devoted to the construction of polynomial 2-surfaces which possess a polynomial area element. In particular we study these surfaces in the Euclidean space $\mathbb R^3$ (where they are equivalent to the PN surfaces) and in the Minkowski space $\mathbb R^{3,1}$ (where they provide the MOS surfaces). We show generally in real vector spaces of any dimension and any metric that the Gram determinant of a parametric set of subspaces is a perfect square if and only if the Gram determinant of its orthogonal complement is a perfect square. Consequently the polynomial surfaces of a given degree with polynomial area element can be constructed from the prescribed normal fields solving a system of linear equations. The degree of the constructed surface depending on the degree and the quality of the prescribed normal field is investigated and discussed. We use the presented approach to interpolate a network of points and associated normals with piecewise polynomial surfaces with polynomial area element and demonstrate our method on a number of examples (constructions of quadrilateral as well as triangular patches).
\end{abstract}


\end{frontmatter}

\section{Introduction}\label{sec intro}

Rational surfaces with Pythagorean normal vector fields (PN surfaces)  were introduced by \cite{Po95} as a surface analogy to Pythagorean hodograph (PH) curves defined previously by \cite{FaSa90}. For a survey of shapes with Pythagorean property see e.g. \citep{Fa08} and references therein. It holds that  PH curves in plane and PN surfaces in space considered as hypersurfaces share some common properties, e.g. they both yield rational offsets. Nevertheless there exist lot of significant differences between these classes of rational varieties. For instance, the curves with Pythagorean hodographs were introduced as  planar {\em polynomial} shapes and a compact formula for their description based on Pythagorean triples of polynomials is available. On the other hand, a description of {\em rational} Pythagorean normal vector surfaces reflecting their dual description was revealed first in \citep{Po95} and it is still not known how to specify these formulas to obtain from them the subclass of polynomial PN surfaces. This could be probably one of the reasons why the PN surfaces do not have as many particular applications as the PH curves. Nonetheless, new attempts to study PN surfaces has again begun recently, see \citep{KoKrVi16,LaSiVr16}.

Indeed, when working with PH curves and PN surfaces then focusing only on the rationality of their offsets can conceal other important properties and it does not offer a full overview of their useful features. In the curve case, another (or maybe the main) very important practical application is based on the fact that the parametric speed (or the length element), and thus also the arc length, of polynomial PH curves is polynomial, too. This is important for formulating efficient real time interpolator algorithms for CNC machines. We recall that the interpolators for general NURBS curves are typically computed using Taylor series expansions. Of course, this approach brings truncation errors caused by omitting higher-order terms. When the Pythagorean hodograph curves are applied for describing the tool path, this problem is overcome. The concept of planar polynomial PH curves was generalized also to spatial polynomial PH curves \citep{FaSa94} which are not hypersurfaces anymore and thus we do not construct their offsets as in the plane case. This can be taken as another reason for preferring the polynomiality of the parametric speed over the rationality of their offsets as a main distinguishing property. Later, planar and spatial PH curves were studied also as rational objects \citep{Po95,FaSi11}. However, we would like to emphasize that for rational PH curves their arc length does not have to be expressible as a rational function of the parameter as the integral of a rational function is not rational, in general.

Analogously to the parametric speed and the arc length in the curve case we recall the area element and the surface area for surfaces. Clearly, the area element, and thus also the surface area, of polynomial PN surfaces is polynomial but the surface area of rational PN surfaces is again not rational, in general. This underlines a prominent role of polynomial PN surfaces and shows a more natural relation between polynomial PH curves to polynomial PN surfaces rather then the rationality of their offsets. Moreover, as the {\em curves with the polynomial/rational line element} (i.e., PH  curves) can be defined in any arbitrary dimension, the same holds also for the {\em surfaces with the polynomial/rational area element}, whose special instances the PN surfaces in 3-space are. Unfortunately, there is not known very much about polynomial PN surfaces. As a particular result we can mention the investigation of a remarkable family of cubic polynomial PN surfaces with birational Gauss mapping, which represent a surface counterpart to the planar Tschirnhausen cubic, the simplest planar polynomial PH curve. A full description of these PN surfaces, among which e.g. the Enneper surface belongs, was presented and their properties were thoroughly discussed in \cite{LaVr11}. Recently an approach for a construction of polynomial PN surfaces based on bivariate polynomials with quaternion coefficients was presented by \cite{KoKrVi16}.

As concerns modelling techniques formulated for PN surfaces, in particular the Hermite interpolation schemes by piecewise PN surfaces, there are not many results from this area. One can find a few indirect algorithms for the interpolations with PN surfaces, where `indirect' means that the resulting surfaces become rational PN only after a suitable reparameterization -- we recall e.g. \citep{JuSa00,BaJuKoLa08}; however these must be always followed by non-trivial trimming of the parameter domain. One can find also a few direct algorithms based on the dual approach, which is especially convenient for PN surfaces, see e.g. \citep{PePo96,LaSiVr16}. Nevertheless these approaches produce rational PN surfaces and are inapplicable when polynomial parameterizations are required. As far as we are aware, the algorithm presented in this paper is the first functional and complex method solving the Hermite problem directly (i.e., without a need of any consequent reparameterization) and formulated without a need of envelope formula (necessary when dual approach is used) and thus yielding polynomial parameterizations.
We will show that in our approach the interpolation problem can be always transformed to solving a system of linear equations. In addition, after a slight modification we present an analogous approach for interpolating with polynomial medial surface transforms yielding rational envelopes (so called MOS surfaces), which are further surfaces playing an important role in solving practical problems originated in technical practice.

The remainder of this paper is organized as follows. Section~$2$ recalls some basic facts concerning curves with polynomial/rational length element (PH and MPH curves) and mainly surfaces with polynomial/rational area element (especially PN and MOS surfaces) that are the principal topic of this paper. A certain generalization of the presented ideas to an $n$-dimensional space and to an arbitrary rational $k$-surface is revealed. In Section~3, we present a simple method for describing and generating all polynomial surfaces with polynomial area element. The results are formulated in the simplest possible way to be later easily applicable for formulating functional algorithms for the Hermite interpolation by piecewise polynomial PN/MOS surfaces. This section contains also a theoretical part devoted to the problem of finding relation between the degree of prescribed normal vector fields, the degree of the obtained surfaces and the dimension of the set of solutions. Efficient tools from the commutative algebra, as e.g. syzygy modules, complexes and Hilbert functions, are used to answer the natural questions, important also for the interpolation. In Section~4, the results from the previous parts are applied to a practical problem of Hermite interpolation by piecewise polynomial surfaces with polynomial area element. Simple methods for smooth surface interpolation using polynomial patches with rational offsets in $\R^3$, or using polynomial medial surface transforms in $\R^{3,1}$ yielding rational envelopes are presented and thoroughly discussed. We will show that in our approach the interpolation problem can be always transformed to solving a system of linear equations. The functionality of the designed algorithms is shown on several examples. Finally, we conclude the paper in Section~5.

\section{Preliminary}\label{sec prelim}

We start with PH curves in plane, and consequently we generalize the approach to an $n$-dimensional space and to an arbitrary rational $k$-surface. Especially, we will focus on 2-surfaces in spaces $\R^3$ and $\R^{3,1}$.

\medskip
A parametric curve $\f x(u)=(x_1(u),x_2(u))^\top$ in $\R^2$ is called a \emph{Pythagorean hodograph curve} ({\em a PH curve} for short)
if there exists a rational function $\sigma(u)$ such that it is satisfied
\begin{equation}\label{eq:PHcond}
x_1'(u)^2 + x_2'(u)^2= \sigma(u)^2.
\end{equation}
This means that for PH curves the squared {\em length element}
\begin{equation}
\mathrm{d}s^2=\f{x}'(u)\!\cdot\!\f{x}'(u)\, \mathrm{d}u^2=\|\f{x}'(u)\|^2\, \mathrm{d}u^2,
\end{equation}
where '$\cdot$' is the standard Euclidean inner product, is a perfect square. Hence, these curves can be also denoted as {\em curves with rational length element}.
Furthermore, this approach is applicable for introducing the PH curves in any dimension and one can speak about PH curves (or curves with rational length element)
in an arbitrary space $\R^n$.  It is evident that all polynomial PH curves in any space $\R^n$ possess polynomial arc length $\ell(u)=\int \|\f{x}'(u)\|\, \mathrm{d}u$.

Next, we follow the same approach for 2-surfaces in $\R^3$ (or in $\R^n$, in general). The squared {\em area element} has the form
\begin{equation}
\mathrm{d}A^2=
\begin{array}{|cc|}
\displaystyle  \f x_u \!\cdot\! \f x_u  &
\displaystyle \f x_u \!\cdot\!  \f x_v \\
\displaystyle  \f x_u \!\cdot\! \f x_v  &
\displaystyle  \f x_v \!\cdot\! \f x_v
\end{array}
\,\,\mathrm{d}u^2\mathrm{d}v^2
=(EG-F^2)\,\mathrm{d}u^2\mathrm{d}v^2,
\end{equation}
where $\f x_u={\partial \f{x}}/{\partial u}$, $\f x_v={\partial \f{x}}/{\partial v}$, and $E= \f x_u \!\cdot\! \f x_u$, $F= \f x_u \!\cdot\! \f x_v$, $G= \f x_v \!\cdot\! \f x_v$ are the coefficients of the first fundamental form. Then a parametric surface $\f x(u,v)$ is called a {\em surface with rational area element}
if there exists a rational function $\sigma(u,v)$ such that it is satisfied
\begin{equation}
EG-F^2=\sigma(u,v)^2.
\end{equation}
Again all polynomial surfaces in space $\R^n$ with polynomial area element  possess polynomial surface area $A(u,v)=\int\!\!\int \sqrt{EG-F^2}\, \mathrm{d}u\mathrm{d}v$.

\medskip
For later use, we mention some fundamental facts extending the previous ideas. Let be given a rational parameterization $\f x(\f u):\R^k\rightarrow \R^{p,q}$, where $\f u=(u_1,\ldots,u_k)$, and $\R^{p,q}$ is a real space of dimension $n=p+q$ equipped with the inner product $\langle\,\cdot \,,\cdot\,\rangle$ of signature $(p,q)$ (especially, if $q=0$ we have the standard Euclidean space, if $q=1$ we have the Minkowski space). We consider a system of tangent vectors $\left(\partial \f x(\f u)/\partial u_1,\ldots,\partial \f x(\f u)/\partial u_k\right)$, or $\big(\f x_1(\f u),\ldots,\f x_k(\f u)\big)$ for short, and compute its corresponding {\em Gram determinant} (or {\em Gramian}\/)
\begin{equation}
\Gamma(\f x_1,\ldots,\f x_k)=\mathrm{det}(g_{ij}),\quad \mbox{ where }\, g_{ij}=\big\langle \f x_i(\f u),\f x_j(\f u) \big\rangle,\quad i,j=1,\ldots,k.
\end{equation}

As known the Gram determinant of given $k$ vectors is equal to the square of the $k$-dimensional volume of the parallelotope spanned by these $k$ vectors. Hence
the squared {\em volume element} has the form
\begin{equation}
\mathrm{d}V^2=\Gamma(\f x_1,\ldots,\f x_k)\,\,\mathrm{d}u_1^2\cdots\mathrm{d}u_k^2.
\end{equation}
To sum up, $\f x(\f u)$ is called a {\em $k$-surface with rational volume element}\/ if there exists a rational function $\sigma(\f u)\in \R(\f u)$ such that
\begin{equation}\label{eq:PG}
\Gamma(\f x_1,\ldots,\f x_k)=\sigma^2(\f u).
\end{equation}

In particular, if $k=1,p=n,q=0$ then \eqref{eq:PG} describes (Euclidean) Pythagorean hodograph curves. For $k=1,p=n-1,q=1$ we obtain Minkowski Pythagorean hodograph (MPH) curves.
If $k=2,p=3,q=1$ then we get the so called MOS surfaces, i.e., medial surfaces obeying a certain sum of squares condition. Finally, when $k=n-1,p=n,q=0$ we arrive at
(Euclidean) hypersurfaces with  rational volume element. As in the curve and surfaces case, a special role is played by polynomial varieties with polynomial volume element as they possess polynomial volume $V(\f u)=\int\!\!\int\!\cdots\!\int \sqrt{\mathrm{det}(g_{ij})}\,\, \mathrm{d}u_1\cdots\mathrm{d}u_k$.

Moreover, as it holds for a hypersurface $\f x(\f u)$
\begin{equation}\label{eq:PN}
\Gamma(\f x_1,\ldots,\f x_{n-1})=\left\| \f x_1\times\cdots\times \f x_{n-1} \right\|^2,
\end{equation}
where $\f x_1\times\cdots\times \f x_{n-1} $ is the generalized cross product providing a normal vector $\f n$, condition \eqref{eq:PG} yields in this case hypersurfaces with \emph{Pythagorean normals} (shortly {\em PN hypersurfaces}\/) in $\R^n$. Their distinguishing property is that they admit two-sided rational $\delta$-offset hypersurfaces
\begin{equation}\label{PNsurf}
\f{x}_{\delta}=\f{x}\pm\delta\frac{\f{n}}{\|\f n\|}=\f{x}\pm\delta\,\frac{\f{x}_1\times\cdots\times\f{x}_{n-1}}{\sigma}.
\end{equation}
It holds that planar PH curves (i.e., curves with rational length element) in $\R^2$ are PN curves (i.e., rational offset curves), and surfaces with rational area element in $\R^3$ are PN surfaces (i.e., rational offset surfaces).

\medskip
As concerns the formulas for rational/polynomial $k$-surfaces with rational volume element (suitable e.g. for formulating interpolation algorithms), these are known only in special cases.
For instance, it was proved in \citep{FaSa90,Ku72} that the coordinates of hodographs
of polynomial planar PH curves and
$\sigma(t)$ form the following Pythagorean triples
\begin{equation}\label{planarPH}
\begin{array}{rcl}
x_1'(u) & = & k(t) \bigl( a^2(u) - b^2(u) \bigr), \\
x_2'(u) & = & 2 k(u) a(t) b(u),\\
\sigma(u) & = & k(u) (a^2(u) + b^2(u)),
\end{array}
\end{equation}
where $a(u)$, $b(u)$, $k(u)\in\mathbb{R}[u]$ are any non-zero polynomials
and $a(u), b(u)$ are relatively prime. The parameterization of the PH curve is then
obtained by integrating the hodograph coordinates from~\eqref{planarPH}. Obviously this approach cannot be used for rational planar PH curves as the integral of a rational function is not rational, in general. Analogous formulas, derived using a similar approach, were found by \cite{FaSa94} for polynomial PH curves in $\R^3$ and by \cite{Mo99} for polynomial MPH curves in $\R^{2,1}$. Later, formulas describing rational PH curves in $\R^3$ and rational MPH curves in $\R^{2,1}$ were presented in \citep{FaSi11,KoLa10}.

The next $k$-surfaces with rational volume element for which compact formulas exist are PN hypersurfaces. In this case, the construction is based on their dual representation. Any rational PN hypersurface can be represented as the envelope of its tangent hyperplanes
\begin{equation}\label{tangents}
\f n(\f u)\cdot \f x=h(\f u),
\end{equation}
where $\f n(\f u)$ is a polynomial normal vector field such that $||\f n(\f u)||^2$ is a perfect square, see \citep{DiHoJu93}, and $h(\f u)$ is a rational function.
Differentiating \eqref{tangents} with respect to $u_i$ gives the system of $n$ linear equations in variables $x_i$
\begin{equation}\label{PN_param}
\f M\cdot \f x =\f H,\qquad\mbox{ where}\quad
\f M =
\left(\f n, \ldots,\frac{\partial \f n}{\partial u_i},\ldots\right)^\top \,\mbox{ and }\,
\f H =
\left(h, \ldots,\frac{\partial h}{\partial u_i},\ldots\right)^\top.
\end{equation}
Solving \eqref{PN_param} we arrive at a general representation $\f x (\f u)= \f M^{-1}\f H$ of PN hypersurfaces with non-degenerate Gaussian image, cf. \citep{Po95}; PN hypersurfaces with degenerate Gaussian image for which $\f M$ is non-invertible, e.g. developable surfaces in $\R^3$, need  a special treatment.

\section{Two remarkable classes of polynomial surfaces with polynomial surface area element}

The method discussed in the previous section, formulated for rational PN hypersurfaces, is not suitable for computing parameterizations of polynomial PN surfaces, coinciding with the class of surfaces with polynomial area element in $\R^3$. And polynomial MOS surfaces as 2-surfaces with polynomial area element in 4-dimensional space $\R^{3,1}$ are not hypersurfaces, thus the presented dual approach cannot be applied inherently. So in what follows, we will reveal another method for describing polynomial surfaces with polynomial area element.

\medskip
When studying varieties with polynomial volume elements then it is sometimes more convenient to prescribe the tangents space (e.g. in case of spatial PH or MPH curves) and sometimes it is more convenient to start with the normal space (e.g. in case of PN surfaces). In the following subsection we will show that both ways are equivalent and thus one can always choose
an approach which is computationally more accessible.

%

\subsection{Gram determinants of $k$-parametric families of vector subspaces and their orthogonal complements}

Consider a set of $0<m<n=p+q$ parameterizations of polynomial vector fields $\R^k\rightarrow\R^{p,q}$ given by $\f u=(u_1,\dots,u_k)\mapsto\f v_i(\f u)$ for $i=1,\dots,m$. Assuming that for almost all $\f u$ the corresponding vectors are linearly independent, we may understand the $m$--tuple $(\f v_1,\dots,\f v_m)$ as a $k$--parametric family of $m$-dimensional subspaces $V(\f u)$. Define the~\emph{reduced Gram determinant}  $\Gamma_0(\f v_1,\dots,\f v_m)$ to be a square-free part of the Gram determinant $\Gamma(\f v_1,\dots,\f v_m)=\det(\langle\f v_i,\f v_j\rangle)_{i,j=1}^m$.

\begin{lemma}
  Let $(\f v_1(\f u),\dots,\f v_m(\f u))$ and $(\f v'_1(\f u),\dots,\f v'_m(\f u))$ be two parameterizations of the same $V(\f u)$. Then there exists a non-zero constant $c\in\R$ such that $\Gamma_0(\f v_1,\dots,\f v_m)=c\cdot \Gamma_0(\f v'_1,\dots,\f v'_m)$,
\end{lemma}

\begin{proof}
Let $A(\f u)$ be a change-of-basis matrix such that $A(\f u)\f v_i(\f u)=\f v'_i(\f u)$. Then the Gram determinants are linked by the relation
\begin{equation}
\det(\langle\f v'_i,\f v'_j\rangle)_{i,j=1}^m = \left(\mathrm{det}(A)\right)^2 \cdot \det(\langle\f v_i,\f v_j\rangle)_{i,j=1}^m
\end{equation}
Since $\left(\mathrm{det}(A)\right)^2$ is a square it is omitted when taking the square-free part. Thus the reduced Gram determinants may differ only by a constant.
\end{proof}

Hence, the quantity $\Gamma_0(\f v_1,\dots,\f v_m)$ does not depend on a particular parametrization and we may define the reduced Gram determinant $\Gamma_0(V(\f u))$ of the $k$-parametric set of $m$-subspaces $V(\f u)=\mathrm{span}\{\f v_1,\dots,\f v_m\}$.

\smallskip
Recall now that for a subspace $V\subset\R^{p,q}$ the totally orthogonal subspace (or the orthogonal complement)  $V^\perp$ is defined as the set of all vectors from $\R^{p,q}$ orthogonal to all vectors of $V$.

\begin{lemma}\label{gramiany}
$\Gamma_0(V(\f u))=c\cdot \Gamma_0(V^\perp(\f u))$.
\end{lemma}

\begin{proof}
To prove this lemma we will use the tools from exterior algebra. Let $\f a=\f a_1\wedge\dots\wedge\f a_k$ and $\f b=\f b_1\wedge\dots\wedge\f b_k$ be two $k$-vectors. We recall that  $\wedge$ is the {\em exterior product}.  Next, the product $\langle\cdot,\cdot\rangle$ on $\R^{p,q}$ induces the product on the $k$-th exterior power $\bigwedge^k(\R^{p,q})$ via the relation
\begin{equation}
\langle\f a,\f b\rangle = \det(\langle \f a_i,\f b_j \rangle )_{i,j,=1}^k.
\end{equation}
Recall that for $n=p+q$ the {\em Hodge star operator} is the isomorphism $\star:\bigwedge^k(\R^{p,q})\rightarrow\bigwedge^{n-k}(\R^{p,q})$ fulfilling
\begin{equation}
\f a\wedge(\star\f b)=\langle\f a,\f b\rangle \f I,
\end{equation}
where $\f a,\f b$ are as above and $\f I$ is the normalized $n$--vector. For $\f a=\f a_1\wedge\dots\wedge \f a_k$ the $\star\f a$ can be written as $\tilde{ \f  a}_1\wedge\dots\wedge\tilde {\f a}_{n-k}$ where $\{\tilde{ \f a}_i\}$ is the basis of the subspace totally orthogonal to the one spanned by $\f a_i$'s.

Now, let $V(\f u)$ be spanned by the vectors $\f v_i(\f u)$ and set $\f v(\f u)=\f v_1(\f u)\wedge\dots\wedge \f v_k(\f u)$. We have
\begin{equation}\label{eq gr V}
\det(\langle \f v_i(\f u),\f v_j(\f u) \rangle )_{i,j=1}^k\f I= \f v(\f u)\wedge\star\f v(\f u).
\end{equation}
Since there exists a basis $\f w_i(\f u)$ of $V^\perp(\f u)$ such that $\star\f v(\f u)=\f w_1(\f u)\wedge\dots\wedge\f w_{n-k}(\f u)$ we may write
\begin{equation}\label{eq gr Vperp}
\det(\f w_i(\f u),\f w_j(\f u))_{i,j=1}^{n-k}=\star\f v(\f u)\wedge\star\star\f v(\f u)=(-1)^{q\,\text{mod}\,2}\f v(\f u)\wedge\star\f v(\f u),
\end{equation}
where we have used the formulas  $\star\star\f v=(-1)^{k(n-k)}(-1)^{q\,\text{mod}\,2}\f v$ and $\star\f v\wedge\f v=(-1)^{k(n-k)}\f v\wedge\star\f v$.
Comparing \eqref{eq gr V} with \eqref{eq gr Vperp} we see that the reduced Gram determinants of subspaces differ only by multiplication of a non-zero constant.
\end{proof}

\subsection{Polynomial PN surfaces in $\R^3$}\label{poly PN}

We recall Lemma~\ref{gramiany} and reformulate the statement for 2-surfaces in 3-dimensional space. Let be given a polynomial parameterized surface $\mathbf x(u,v)$. Consider the tangent space $V(u,v)=\mathrm{span}\{\f x_u(u,v),\f x_v(u,v)\}$ and the normal space $V^\perp(u,v)=\mathrm{span}\{\f n(u,v)\}$. Then it holds
\begin{equation}\label{condPN}
\Gamma(\f x_u,\f x_v)=f^2\,\Gamma(\f n),
\end{equation}
where $f(u,v)\in\R(u,v)$ is a non-zero factor.

Thus when looking for some parameterized polynomial PN surface (polynomial surface with polynomial surface element in $\R^3$) it is natural to start with a polynomial normal vector field $\f n(u,v)$ of degree $k$ such that $||\f n(u,v)||^2$ is a perfect square. Its parameterization can be easily gained from polynomial Pythagorean quadruples, cf. \citep{DiHoJu93}. By \eqref{condPN} the Pythagorean property of $\f n(u,v)$ guarantees the polynomiality of area element.

In addition, to determine an associated polynomial PN parameterization of degree $\ell+1$ in a direct way, we have to find suitable polynomial vector fields
\begin{equation}\label{PN_PuPv}
\begin{array}{c}
\displaystyle
\f q(u,v)=
\left(
\sum_{i+j\leq\ell}
\mbox{\hspace*{-0ex}}q_{1ij}{u^iv^j},
\sum_{i+j\leq\ell}
\mbox{\hspace*{-0ex}}q_{2ij}{u^iv^j},
\sum_{i+j\leq\ell}
\mbox{\hspace*{-0ex}}q_{3ij}{u^iv^j}
\right)^{\top},\\[4ex]
\displaystyle
\f r(u,v)=
\left(
\sum_{i+j\leq\ell}
\mbox{\hspace*{-0ex}}r_{1ij}{u^iv^j},
\sum_{i+j\leq\ell}
\mbox{\hspace*{-0ex}}r_{2ij}{u^iv^j},
\sum_{i+j\leq\ell}
\mbox{\hspace*{-0ex}}r_{3ij}{u^iv^j}
\right)^{\top},
\end{array}
\end{equation}
which will play the role of $\f x_u$, $\f x_v$, respectively. Thus, $\f q$, $\f r$ must satisfy the following conditions
\begin{equation}\label{eq PN soustava}
\begin{array}{rcl}
\f q \cdot \f n & \equiv & 0,\\
\f r \cdot \f n & \equiv & 0,\\[1ex]
\displaystyle \frac{\partial\f q}{\partial v} - \displaystyle \frac{\partial\f r}{\partial u} & \equiv & 0,
\end{array}
\end{equation}
where the third equation expresses the condition for the~integrability.  Since a polynomial of degree $n$ in two variables possesses $\binom{n+2}{2}$ coefficients, the problem is now transformed to solving a system of $2\binom{k+\ell+2}{2}+3\binom{\ell+1}{2}$ homogeneous linear equations with $6\binom{\ell+2}{2}$ unknowns $q_{1ij}, q_{2ij}, q_{3ij},r_{1ij},r_{2ij},r_{3ij}$. The corresponding PN parameterization is then obtain as
\begin{equation}\label{Eq:integrace}
\f x(u,v)=\int\f {q}(u,v)\,\mathrm{d}u+\f {c}(v),
\mbox{ where}
\qquad
\f{c}(v)=\left[\int \f{r}(u,v)\,\mathrm{d}v-\int \f{q}(u,v)\,\mathrm{d}u\right]_{u=0}.
\end{equation}

For $\ell$ large enough, system of equations \eqref{eq PN soustava} is solvable. In this case we arrive at a polynomial PN parameterization such that $\f x_u \times \f x_v=f(u,v)\f n(u,v)$, where $f(u,v)$ is a factor balancing suitably the degrees of $\f n$ and $\f x$. We can formulate

\begin{proposition}\label{PN_howto}
Given in $\R^3$ a polynomial vector field $\f n(u,v)$ such that $||\f n(u,v)||^2$ is a perfect square. Then there exists a polynomial PN surface, i.e., a polynomial surface with polynomial surface area element, possessing $\f n(u,v)$ as its normal vector field.
\end{proposition}

When computing $\f q(u,v)$ orthogonal to a given normal field of degree $k$ it is always necessary to prescribe first a suitable degree $\ell$ for which we have guaranteed the existence of $\f q$. This degree is of course in a direct relation to the dimension of the solution, which depends on the number of equations and unknowns. From this reason we will study the independence of the linear equations in the system.


For the normal field  $\f n(u,v)=(n_1(u,v),n_2(u,v),n_3(u,v))$ the set of all vector fields $\f q=(q_1,q_2,q_3)\in\R^3[u,v]$ orthogonal to $\f n$ forms a module over the ring $\R[u,v]$. This is called a \emph{syzygy module}, i.e.,
\begin{equation}
  \mathrm{Syz}(\f n)=\{\f q \in \R^3[u,v]\mid\, \f q\cdot\f n\equiv 0\}.
\end{equation}

\begin{theorem}\label{thm free module}
  The $\mathrm{Syz}(\f n)$ is a free module of rank two.  Moreover, two vector fields $\f q(u,v)$ and $\f r(u,v)$ form its basis if and only if there exists a constant $c\in\R$ such that $\f q(u,v)\times\f r(u,v)=c\,\f n(u,v)$.
\end{theorem}

\begin{proof}
As the particular steps of the proof would directly follow the ideas and results on syzygies of four polynomials in two variables from \citep{CheCoLi05} we omit it and refer the readers to the mentioned paper.

\end{proof}

\begin{example}\label{ex mu basis}
  Let $\f n=(2u,2v,1-u^2-v^2)$ be a polynomial normal field related to a parameterization of the unit sphere.  It can be easily verified that the two vector fields
  \begin{equation}
 \f q=(u^2-1,uv,2u) \quad\text{and}\quad\f r=(uv, v^2-1,2 v)
  \end{equation}
  fulfils $\f q\times\f r=\f n$ and thus they form a basis of $\mathrm{Syz}(\f n)$. In other words any polynomial vector field $\f p$ orthogonal to $\f n$ can be uniquely written as $\f p = a \, \f q+ b \, \f r$ for some polynomials $a,b\in\R[u,v]$.
\end{example}

\begin{remark}\rm
Let us demonstrate in more detail the main added value of knowing a basis of $\mathrm{Syz}(\f n)$, i.e., that any vector field orthogonal to $\f n$ can be uniquely generated as an algebraic combination  of this basis. In this situation, one does not have to consider (and thoroughly discuss) situations when polynomial vector fields are obtained from generating set using rational functions as multiplying coefficients. Moreover, then the fundamental question must read: Which rational coefficients yield polynomial combinations? We recall e.g. Section 4.1 in \cite{KoKrVi16} in which the generating set $\big\{ (-u^2+v^2+1,-2uv,2u), (2uv,-u^2+v^2+1,-2v) \big\}$ (not being a basis) is used for determining cubic polynomial PN surfaces applying particular quadratic rational functions.
\end{remark}

Let $\f n(u,v)$ be a polynomial vector field of degree $k$ and in addition assume $\gcd(n_i)=1$. Then there exist only finitely many points $(u,v)$ such that $n_i(u,v)=0$, for $i=1,2,3$. These points are called \emph{base points} of the vector field. The consecutive result depends on the existence of such base points, which must be considered over $\C$ and also at infinity
(i.e., common roots of the terms of $n_i(u,v)$ of degree $k$).

\begin{lemma}\label{thm jakkoliv}
The system of linear equations $\f q \cdot \f n \equiv 0$ has the full rank if and only if
$\f n$ is basepoint-free. Then the dimension of the~set of vector fields $\f q$ of degree at most $\ell$ orthogonal to $\f n$ is equal to
\begin{equation}
\Lambda(\ell,k):= 3{{\ell+2}\choose{2}}-{{k+\ell+2}\choose{2}}
\end{equation}

\end{lemma}

\begin{proof}
Because of its technical nature the proof is postponed to Appendix.
\end{proof}

\begin{lemma}
The system of linear equations $q_v\equiv p_u$ has the full rank. Thus, the dimension of the~set of pairs $(q,r)$ of compatible polynomials of degree at most $\ell$ is equal to
\begin{equation}
\Omega(\ell) = {{\ell+3}\choose{2}}-1.
\end{equation}
\end{lemma}

\begin{proof}
This problem can be directly transformed to computing the non-absolute coefficients of a polynomial of degree $\ell+1$ since after the computation of its partial derivatives one immediately obtains pairs of compatible polynomials of degree $\ell$.
\end{proof}

To sum up, if we have prescribed a polynomial normal vector field $\f n(u,v)$ of degree $k$ then the family of polynomial parameterized surfaces $\f x(u,v)$ of degree $\ell + 1$ (up to translation) with  $\f n(u,v)$ as its normal vector field has the dimension
\begin{equation}\label{Eq:odhad}
2 \Lambda(\ell,k) + 3 \Omega(\ell) - 6 {{\ell+2}\choose{2}} + \Delta,
\end{equation}
where $\Delta \geq 0$ represents a correction responsible for the quality of the normal vector field. For instance $\Delta > 0$ for $\f n(u,v)$ possessing base points (which is a typical property of parameterizations of sphere-like surfaces, used in this paper).

Let us emphasize a main importance of \eqref{Eq:odhad} for practical applications studied in this paper. Although $\Delta$ is generally difficult to compute, \eqref{Eq:odhad} immediately reveals a clear effect, i.e., we can easily find to any prescribed $k$ the upper bound for $\ell$. Moreover, as non-standard vector fields $\f n(u,v)$ lead to solutions with more free parameters, we are usually able (especially for Pythagorean normal vector fields) to construct parameterizations $\f x(u,v)$ of lower degree than the computed upper bound, see Example~\ref{Exmp:PNexmp1}.


\smallskip
The previous observations and equations \eqref{eq PN soustava} will be used later for formulating an algorithm for the Hermite interpolation by piecewise polynomial PN surfaces.

\begin{example}\label{Exmp:PNexmp1}\rm
Consider the normal vector field $\f n(u,v)=(2u,2v,u^2+v^2-1)^\top$. Using \eqref{Eq:odhad} we have guaranteed that linear equations \eqref{eq PN soustava} possess a solution for $\ell \geq 3$. However Pythagorean normal vector field $\f n(u,v)$ has base points and thus we obtain a 3-parametric solution already for quadratic polynomials \eqref{PN_PuPv}. In particular, we arrive at the following family of PN surfaces (up to translation)
\begin{multline}
\f x(u,v) = \left(\frac{1}{3} \lambda _1 u \left(u^2+3 v^2-3\right)-\frac{1}{3} \lambda _2 v \left(-3 u^2+v^2+3\right)+\lambda _3 u \left(-u^2+v^2+3\right), \right. \\
\left. \frac{2}{3} \lambda _1 v \left(v^2-3\right)-\frac{1}{3} u \left(\lambda _2 \left(u^2-3 v^2+3\right)+6 \lambda _3 u v\right),  u \left(3 \lambda _3 u-2 \lambda _2 v\right)-\lambda _1 \left(u^2+2 v^2\right)\right)^\top
\end{multline}
with the area element equal to
\begin{equation}
\sigma(u,v)^2 = f(u,v)^2 (u^2+v^2+1)^2,
\end{equation}
where
\begin{equation}
f(u,v)^2 = \left[ -6 \lambda _3^2 u^2+2 \lambda _1 \left(\lambda _3 \left(u^2-v^2-3\right)-2 \lambda _2 u v\right)-\lambda _2^2 \left(u^2+v^2+1\right)+4 \lambda _2 \lambda _3 u v-2 \lambda _1^2 \left(v^2-1\right) \right]^2.
\end{equation}
This also confirms the result from paper \citep{LaVr11} in which the polynomial cubic surfaces were thoroughly investigated and the same three generating surfaces were found.
\end{example}

\begin{remark}\rm
In addition, we would like to stress that a non-constant factor $f(u,v)$ indicates the existence of a curve on the surface $\f x(u,v)$ where the normal field vanishes. In Fig.~\ref{Fig:singular}, the cubic PN surface from Example~\ref{Exmp:PNexmp1} (for chosen values $\lambda_1 = \lambda_3=1$ and $\lambda_2 = 1/10$) with such a curve is shown.
\end{remark}

\begin{figure}[tb]
\begin{center}
\psfrag{X}{$\f{x}(u,v)$}
\psfrag{0}{$0$}
\psfrag{1}{$1$}
\psfrag{u}{$u$}
\psfrag{v}{$v$}
\psfrag{u0}{$v=-2$}
\psfrag{v1}{$u=-2$}
\psfrag{u2}{$v=2$}
\psfrag{v2}{$u=2$}
\psfrag{f}{$f(u,v)=0$}
\psfrag{g}{}
\includegraphics[width=0.85\textwidth]{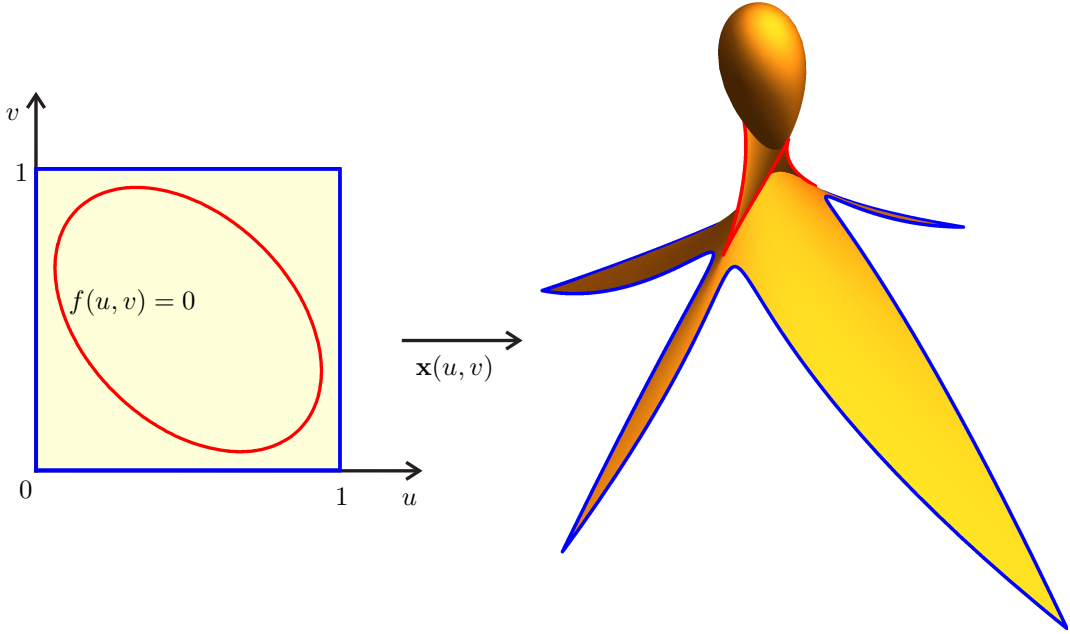}
\begin{minipage}{0.9\textwidth}
\caption{Left: Parametric domain with the ellipse (red) given by the factor $f(u,v)=0$. Right: The cubic surface $\f x(u,v)$ with the curve (red) on it corresponding to $f(u,v)=0$ at which points the normal vector field vanishes.}\label{Fig:singular}
\end{minipage}
\end{center}
\end{figure}

\subsection{Polynomial MOS surfaces in $\R^{3,1}$}\label{poly MOS}

MOS surfaces, i.e., \underline{M}edial surfaces \underline{O}beying the \underline{S}um of squares condition, were introduced by \cite{KoJu07} as a surface analogy of MPH curves in four-dimensional Minkowski space $\mathbb{R}^{3,1}$. The distinguishing property of MOS surfaces is that if considered as an MST (medial surface transform) of a spatial domain, the associated envelope and its offsets admit exact rational parameterization.

For the sake of brevity, we recall at least an expression of the envelope associated to a medial surface transform $\f{x}(u,v)=(x,y,z,r)^{\top}(u,v)$ in $\R^{3,1}$. If we denote by $\hat{\f x}(u,v) = (x,y,z)^{\top}(u,v)$ the corresponding medial surface in $\R^{3}$ then the closed-form {\em envelope formula} has the form
\begin{equation}\label{obalka_MST2}
\f{b}^{\pm}(u,v)=\hat{\f x}(u,v)-r \f n ^{\pm}(u,v),
\end{equation}
where
\begin{equation}\label{cond_MOS}
\begin{array}{rcl}
\f n^{\pm} & = &
\displaystyle \frac{1}{{\hat{E}}{\hat{G}}-{{\hat{F}}}^2}
\left[\left(\frac{\partial r}{\partial u}{\hat{G}}-\frac{\partial r}{\partial v}{\hat{F}}\right){\hat{\f x}}_u+\left(\frac{\partial r}{\partial v}{\hat{E}}-\frac{\partial r}{\partial u}{\hat{F}}\right){\hat{\f x}}_v\mp\sqrt{{E}{G}-{F}^2}({{\hat{\f x}}_u}\times {{\hat{\f x}}_v})\right],
\end{array}
\end{equation}
where $\f n^{\pm}$ is a unit vector perpendicular to $\f b^{\pm}$. The components ${E},{F},{G}$ of the first fundamental form of ${\f x}(u,v)$ are computed using the indefinite Minkowski inner product with the signature $(3,\! 1)$, whereas the components $\hat{E},\hat{F},\hat{G}$ of the first fundamental form of $\hat{\f x}(u,v)$ are determined using the standard Euclidean inner product in $\R^3$. Then MOS surfaces are rational surfaces characterized by the condition
\begin{equation}\label{eq:MOScondition}
{E}{G}-{F}^2=\sigma^2(u,v),\quad \mbox{ where }\,  \sigma(u,v)\in\R(u,v),
\end{equation}
that guarantees the rationality of \eqref{cond_MOS} and thus of the envelope $\f{b}^{\pm}(u,v)$. From this it is evident that MOS surfaces are simultaneously surfaces with rational area element in $\R^{3,1}$.

If points in the projective closure of $\R^{3,1}$ are described using the standard homogeneous coordinates $(x_0:x_1:x_2:x_3:x_4)$ then the equation $x_0=0$ describes the {\em ideal hyperplane} as the set of all asymptotic directions, i.e., of points at infinity. The subset of the ideal hyperplane which is invariant with respect to transformations maintaining Minkowski inner product (i.e., {\em Lorentz transforms}\/) is called the {\em absolute quadric} $\Sigma$ and characterized by
\begin{equation}\label{abskvadr}
\Sigma:\,\, x_1^2+x_2^2+x_3^2-x_4^2 = x_0=0.
\end{equation}

Now, consider in $\R^{3,1}$ a surface given by the parametrization $\f x(u,v)$. At regular points (i.e., where the vectors $\f x_u, \f x_v$ are linearly independent), the {\em normal vectors} of $\f x$ (vectors orthogonal to the tangent 2-plane $\tau(u,v)$ with respect to Minkowski inner product, cf.~Fig.~\ref{norm_bivec}) satisfy the two linear equations
\begin{equation}\label{eq:88}
\begin{array}{c}
\langle \f n,\f x_u\rangle \equiv 0,\\
\langle \f n,\f x_v\rangle \equiv 0.
\end{array}
\end{equation}
Among them, the {\em isotropic normal vectors} are described by
\begin{equation}\label{eq:95}
\langle\f n,\f n \rangle \equiv 0.
\end{equation}
As shown in \citep{BaJuKoLa10}, these isotropic normal vectors of $\f x(u,v)$ have the form \eqref{cond_MOS} and play a significant role in the {\em envelope formula} \eqref{obalka_MST2}.

The isotropic normals $\f n^{\pm}$ can be identified with points of the oval quadric \eqref{abskvadr} considered as the unit sphere in $\R^3$.
For each point $\f x(u,v)$ we obtain two isotropic normal vectors $\f{n}^{\pm}$, which correspond to two points on $\Sigma$ obtained as intersection of the line conjugated with the ideal line of $\tau(u,v)$ with respect to $\Sigma$. The set of these points forms two components $\cal G^{\pm}$, which is usually called the {\em isotropic Gauss image} of $\f x (u,v)$, cf.~\citep{BaJuKoLa10}

\begin{figure}[tb]
\begin{center}
{\psfrag{A}{ }
\psfrag{a}{$\f{x}(u,v)$}
\psfrag{au}{$\f{x}_u$}
\psfrag{av}{$\f{x}_v$}
\psfrag{T}{$\tau(u,v)$}
\psfrag{ilt}{}
\psfrag{n1}{$\f{n}_1$}
\psfrag{n2}{$\f{n}_2$}
\psfrag{pi}{$\pi_{\infty}:x_0=0$}
\includegraphics[width=10cm]{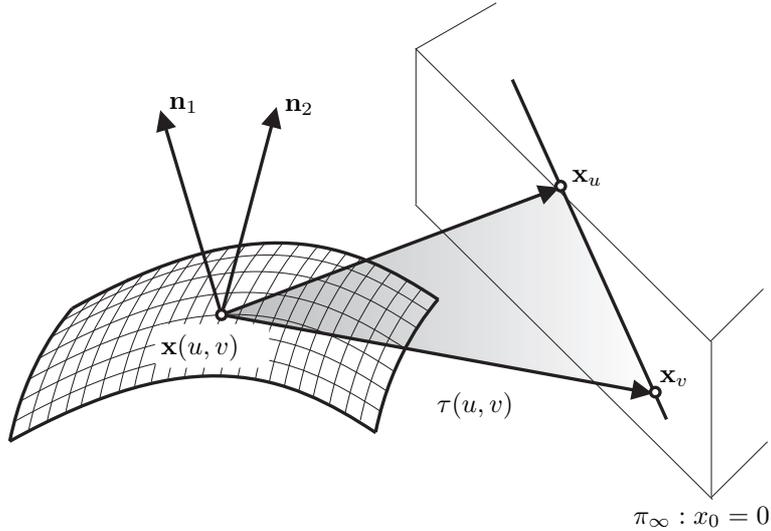}}
\begin{minipage}{0.9\textwidth}
\caption{Normal plane $\nu=\mathrm{span}\{\f{n}_1,\f{n}_2\}$ at $\f x(u,v)$ and the ideal line of the corresponding tangent 2-plane $\tau(u,v)$.}\label{norm_bivec}
\end{minipage}
\end{center}
\end{figure}

\medskip
To find a method for deriving parameterizations of polynomial MOS surfaces, later applicable for Hermite interpolation, we use the approach that worked before for PN (hyper)surfaces in $\R^3$. Firstly, we again recall Lemma~\ref{gramiany} and reformulate the statement for 2-surfaces in 4-dimensional space.
Let be given a polynomial parameterized surface $\mathbf x(u,v)$ in 4-dimensional space. Consider the tangent space $V(u,v)=\mathrm{span}\{\f x_u(u,v),\f x_v(u,v)\}$ and the normal space $V^\perp(u,v)=\mathrm{span}\{\f n_1(u,v),\f n_2(u,v)\}$. Then it holds
\begin{equation}
\Gamma(\f x_u,\f x_v)=f^2\, \Gamma(\f n_1,\f n_2),
\end{equation}
where $f(u,v)\in\R(u,v)$ is a non-zero factor.

This means that it is again possible to start with suitable normal vectors when constructing parameterizations of polynomial  MOS surfaces as the condition on the polynomiality of the area element depends on $\Gamma(\f n_1,\f n_2)$. Clearly, when at least one of the normal vectors $\f n_1$, or $\f n_2$ is isotropic, i.e., its squared norm is zero, then  $\Gamma(\f n_1,\f n_2)$ is automatically a perfect square.

Therefore after a slight modification we can use the main ideas from the approach discussed in the previous section. We start with the normal space
$\mathrm{span}\{\f n^+(u,v),\f n^-(u,v)\}$ given by the polynomial isotropic vectors of degree $k$, i.e., $\langle \f n^{\pm},\f n^{\pm} \rangle  \equiv  0$. Their parameterizations can be again obtained from polynomial Pythagorean quadruples, cf. \citep{DiHoJu93}.

To determine an associated polynomial MOS parameterization of degree $\ell +1$, we are supposed to find suitable polynomial vector fields
\begin{equation}\label{MOS_PuPv}
\begin{array}{c}
\displaystyle
\f q(u,v)=
\left(
\sum_{i+j\leq\ell}
\mbox{\hspace*{-0ex}}q_{1ij}{u^iv^j},
\sum_{i+j\leq\ell}
\mbox{\hspace*{-0ex}}q_{2ij}{u^iv^j},
\sum_{i+j\leq\ell}
\mbox{\hspace*{-0ex}}q_{3ij}{u^iv^j},
\sum_{i+j\leq\ell}
\mbox{\hspace*{-0ex}}q_{4ij}{u^iv^j}
\right)^{\top},
\\[4ex]
\displaystyle
\f r(u,v)=
\left(
\sum_{i+j\leq\ell}
\mbox{\hspace*{-0ex}}r_{1ij}{u^iv^j},
\sum_{i+j\leq\ell}
\mbox{\hspace*{-0ex}}r_{2ij}{u^iv^j},
\sum_{i+j\leq\ell}
\mbox{\hspace*{-0ex}}r_{3ij}{u^iv^j}
\sum_{i+j\leq\ell}
\mbox{\hspace*{-0ex}}r_{4ij}{u^iv^j}
\right)^{\top},
\end{array}
\end{equation}
which will play the role of $\f x_u$, $\f x_v$, respectively. Thus, $\f q$, $\f r$ must satisfy the following conditions
\begin{equation}\label{eq MOS soustava}
\begin{array}{rcl}
\langle\f q , \f n^{\pm} \rangle& \equiv & 0,\\
\langle\f r, \f n^{\pm} \rangle& \equiv & 0,\\[1ex]
\displaystyle \frac{\partial\f q}{\partial v} - \displaystyle \frac{\partial\f r}{\partial u} & \equiv & 0,
\end{array}
\end{equation}
where the third equation expresses the condition for the~integrability.
For $\ell$ large enough, system of linear equations \eqref{eq MOS soustava} with unknowns $q_{1ij}, q_{2ij}, q_{3ij}, q_{4ij}, r_{1ij},r_{2ij},r_{3ij}, r_{4ij}$ is solvable
and  we arrive at the corresponding MOS parameterization
\begin{equation}
\f x(u,v)=\int\f {q}(u,v)\,\mathrm{d}u+\f {c}(v),
\mbox{ where}
\qquad
\f{c}(v)=\left[\int \f{r}(u,v)\,\mathrm{d}v-\int \f{q}(u,v)\,\mathrm{d}u\right]_{u=0},
\end{equation}
for which it holds $EG-F^2=f(u,v)^2\Gamma(\f n^+,\f n^-)$, where $f(u,v)$ is a factor balancing suitably the degrees of $\f n^{\pm}$ and $\f x$.
Hence, we can formulate

\begin{proposition}\label{MOS_howto}
Given in $\R^{3,1}$ isotropic polynomial vector fields $\f n^+(u,v)$ and $\f n^-(u,v)$. Then there exists a polynomial MOS surface, i.e., a polynomial surface with polynomial surface area element, possessing $\mathrm{span}\{\f n^+(u,v),\f n^-(u,v)\}$ as its normal space.
\end{proposition}

\begin{remark}\rm
Obviously, for generating arbitrary MOS parameterizations it is sufficient when only one of the vectors $\f n_1, \f n_2$ is isotropic. This guarantees  that $\Gamma(\f n_1,\f n_2)$ is a perfect square.  However, for interpolation purposes it is then necessary to ensure the continuity conditions in the other way, cf. Section~\ref{Se:MOS_interpolation}.
\end{remark}

\begin{example}\rm
Consider the isotropic normal vector field $\f n(u,v)=(2u,2v,u^2+v^2-1,u^2+v^2+1)^\top$. Solving \eqref{eq MOS soustava} for linear \eqref{MOS_PuPv} yields $5$-parametric family of quadratic MOS surfaces with the parametric description (up to translation)
\begin{multline}
\f x(u,v) = \left( \frac{1}{2} \lambda _4 \left(v^2-u^2\right)+\lambda _5 u+\lambda _3 u v+\lambda _2 v,-\frac{1}{2} \lambda _3 u^2+\lambda _2 u-\lambda_4 u v+\frac{\lambda _3 v^2}{2}+\lambda _1 v, \right. \\
\left. \frac{1}{2} \left(\lambda _5 u^2+2 \lambda _4 u+2 \lambda _2 u v+\lambda _1 v^2-2 \lambda _3 v\right),-\frac{1}{2} \lambda _5 u^2+\lambda _4 u-\lambda _2 u v-\frac{\lambda _1 v^2}{2}-\lambda _3 v\right)^{\top}.
\end{multline}
And
\begin{equation}
\sigma^2(u,v) = \left[ -\lambda _2^2-\left(\lambda _3^2+\lambda _4^2\right) \left(u^2+v^2\right)+2 \lambda _2 \left(\lambda _3 u-\lambda _4 v\right)-\lambda _5
   \left(\lambda _4 u+\lambda _3 v\right)+\lambda _1 \left(\lambda _5+\lambda _4 u+\lambda _3 v\right) \right]^2.
\end{equation}

\end{example}

\section{Smooth surface interpolation using polynomial patches with polynomial area element}\label{interpol}

In this section we will show how the ideas and results from the previous sections can be directly applied to a practical problem of Hermite interpolation by piecewise polynomial surfaces with polynomial area element. Mainly we will discuss a method for smooth surface interpolation using polynomial patches with rational offsets. Then we sketch in short an analogous approach also for polynomial medial surface transforms yielding rational envelopes.

\subsection{Hermite interpolation by piecewise polynomial surfaces with rational offsets}\label{Se:PN_interpolation}

In what follows we present a {\em direct} method for interpolating given network of position data (points) and first order data (normals) by piecewise {\em polynomial} surfaces with rational offsets (Pythagorean normal surfaces). We start with the construction of one quadrilateral/triangular patch interpolating prescribed corner points and normals, and consequently the approach will be extended also for  $m \times n$ points arranged in a rectangular grid (for more details about quadrilateral mesh generation and processing see  e.g. \cite{BoLePiPuSiTaZo13} and also for smoothly joined triangular patches interpolating triangular meshes, cf. \cite{Fa86}.

\medskip
Consider four points $\f p_{ij}$, $i,j=0,1$, and four associated tangent planes $\tau_{ij}$ determined by the unit normal vectors $\f N_{ij}$ (for quadrilateral patches); or
three points $\f p_{ij}$, $i,j=0,1$ and $i+j<2$, and three associated tangent planes $\tau_{ij}$ determined by the unit normal vectors $\f N_{ij}$ (for triangular patches). Following the ideas presented in the previous sections, we can see that the whole algorithm consists of two subparts: (i) first, a suitable normal vector field $\f n(u,v)$ interpolating data $\f n_{ij}=\lambda_{ij} \f N_{ij}$, $\lambda_{ij}  \in \R$, and having the polynomial norm (i.e., satisfying the Pythagorean property) must be constructed; (ii)~next, a polynomial patch interpolating the points $\f p_{ij}$ and possessing normal vector field $\f n(u,v)$ (which guarantees the PN property) is computed.

\smallskip
As concerns {\sc Part (i)}, any method for interpolating data  $\f N_{ij}$ by a (quadrilateral/triangular) rational patch $\f N(u,v)$ on the unit sphere $\mathcal{S}^2$ can be utilized, see e.g. \citep{AlNeSchu96}. For the sake of completeness and to show the functionality and the simplicity of the designed algorithm, we recall one standard method based on using the stereographic projection. Nonetheless, one significant limitation of this approach should be noted -- the points $\f N_{ij}$ on the unit sphere $\mathcal{S}^2$ must be suitably distributed and the Gauss image of the interpolating surface cannot contain the chosen center of the stereographic projection. This means that in some cases a preliminary coordinate transformation is needed.

So, we choose a suitable center  of the stereographic projection $\pi$ (preferably on the opposite hemisphere; see the limitations mentioned above) and project data $\f N_{ij} \in \mathcal{S}^2$ to the plane $\R^2$. Then, we construct a suitable rational patch in $\R^2$ interpolating $\pi(\f N_{ij})$. For instance, in the {\em quadrilateral case} one can consider the bilinear patch
\begin{equation}\label{BL_patch}
\widehat{\f N}(u,v) = \pi(\f N_{00}) \, (1 - u) (1 - v) + \pi(\f N_{10}) \, u (1 - v) +  \pi(\f N_{11}) \, u v +   \pi(\f N_{01}) \, (1 - u) v, \quad u,v \in [0,1];
\end{equation}
or in the {\em triangular case} one can consider the linear patch
\begin{equation}\label{L_patch}
\widehat{\f N}(u,v) = \pi(\f N_{10}) \,u + \pi(\f N_{01}) \, v +  \pi(\f N_{00}) \, (1-u-v)  \quad u \in [0,1],\, v \in [0,1-u].
\end{equation}

The inverse stereographic projection $\pi^{-1}$ yields a rational patch $\f N(u,v)$ on $\mathcal{S}^2$. In addition, as we are interested not in rational but in polynomial normal vector field $\f n(u,v)$ we can omit the least common denominator and consider only numerators of the parameterization. Thus we arrive at a polynomial parameterization $\f n(u,v)$ of a sphere-like surface, see \citep{AlNeSchu96}, fulfilling the Pythagorean property and moreover satisfying the prescribed interpolation conditions
\begin{equation}
\f n(i,j) = \lambda_{ij} \, \f N_{ij}, \quad \lambda_{ij} \in \R.
\end{equation}

\smallskip
Once we have a suitable polynomial vector field $\f n(u,v)$ of degree $k$ we can continue with {\sc Part (ii)} of the algorithm. Our goal is to find a polynomial patch $\f x(u,v)$ of prescribed degree $\ell+1$ possessing $\f n(u,v)$ as its associated normal vector field and interpolating given position data, i.e., it must hold
\begin{equation}\label{Eq:PNinterpolace_rov1}
\f x_u \cdot \f n \equiv 0, \quad \f x_v \cdot \f n \equiv 0,
\end{equation}
and
\begin{equation}\label{Eq:PNinterpolace_rov2}
\f x(i,j) =  \f p_{ij}.
\end{equation}
Thus for further computations, we prescribe a polynomial surface
\begin{equation}\label{Eq:PNinterpolace_surface}
\f x(u,v)=
\left(
\sum_{i+j\leq\ell+1} x_{1ij}{u^iv^j},
\sum_{i+j\leq\ell+1} x_{2ij}{u^iv^j},
\sum_{i+j\leq\ell+1} x_{3ij}{u^iv^j}
\right)^{\top}
\end{equation}
with its coefficients taken as free parameters to be determined by the above constraints.

Clearly differentiating \eqref{Eq:PNinterpolace_surface} with respect to $u$, $v$ we arrive at $\f q(u,v)$, $\f r(u,v)$, respectively, see \eqref{PN_PuPv}.
Let us emphasize that starting with one polynomial parameterization \eqref{Eq:PNinterpolace_surface} instead of two in \eqref{PN_PuPv} and computing partial derivatives of \eqref{Eq:PNinterpolace_surface} instead of integrating \eqref{Eq:integrace} is more appropriate for the purpose of interpolation as one does not have to take care of the compatibility condition. Moreover, the number of the resulting linear equations is significantly lower. On the other hand, for gaining the theoretical results as e.g. for the estimation of degree of the resulting surface, prescribing two independent parameterizations $\f q(u,v)$, $\f r(u,v)$ from \eqref{PN_PuPv} was more convenient.

To conclude the method, expressions \eqref{Eq:PNinterpolace_rov1} and \eqref{Eq:PNinterpolace_rov2} depend linearly on coefficients $x_{1ij},x_{2ij},x_{3ij}$ of $\f x(u,v)$ and therefore can be rewritten as a system of linear equations, which is easy to solve. Solving the equations from systems \eqref{Eq:PNinterpolace_rov1} and \eqref{Eq:PNinterpolace_rov2} yields a polynomial PN patch interpolating the points $\f p_{ij}$ and touching the planes $\tau_{ij}$ at these points.

%

\begin{figure}[tb]
\begin{center}
\psfrag{x}{{\color{black}$\f{x}(u,v)$}}
\psfrag{p00}{$\f p_{00}$}
\psfrag{p10}{$\f p_{10}$}
\psfrag{p22}{$\f p_{11}$}
\psfrag{p01}{$\f p_{01}$}
\psfrag{n00}{$\f N_{00}$}
\psfrag{n10}{$\f N_{10}$}
\psfrag{n22}{$\f N_{11}$}
\psfrag{n01}{$\f N_{01}$}
\includegraphics[width=0.58\textwidth]{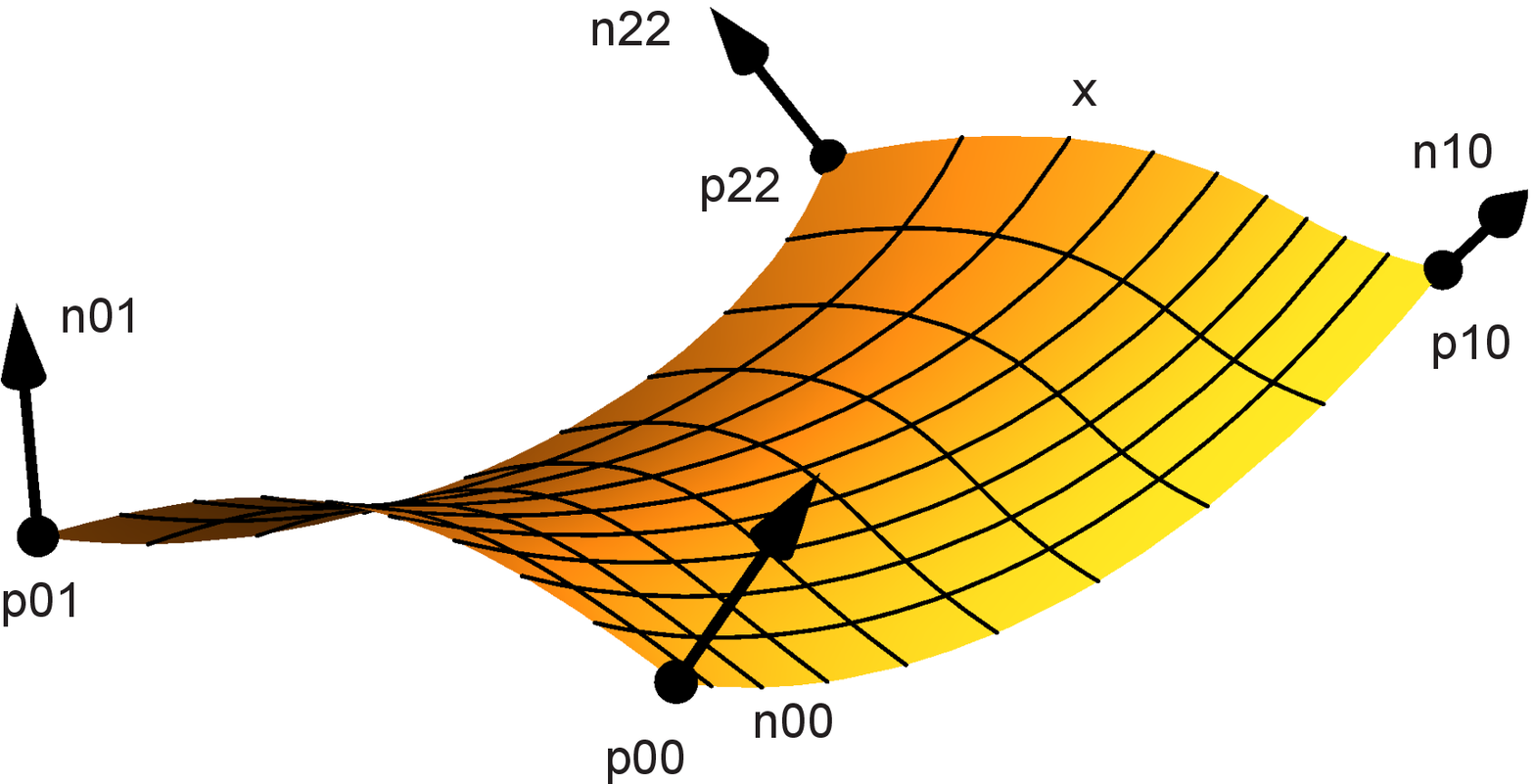}
\hspace{3ex}
\includegraphics[width=0.37\textwidth]{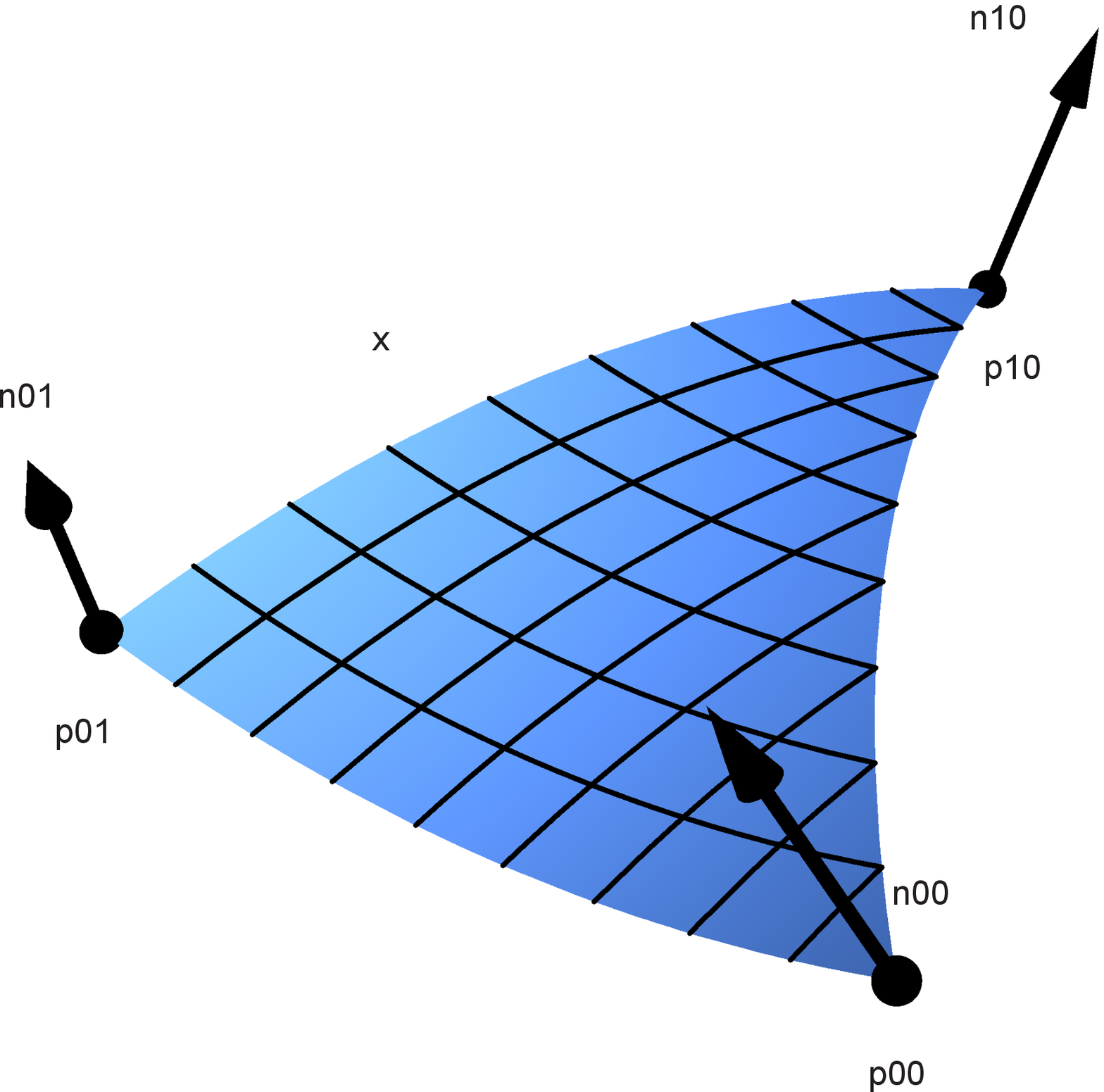}
\begin{minipage}{0.9\textwidth}
\caption{Quadrilateral (left) and triangular (right) polynomial PN patches of degrees $8$ and $4$, respectively, interpolating points $\f p_{ij}$ and possessing tangent planes at $\f p_{ij}$ given by unit normal vectors $\f N_{ij}$  from Examples~\ref{Exmp:1PNpatche} and \ref{Exmp:1PNtriangle}.}\label{Fig:1patch}
\end{minipage}
\end{center}
\end{figure}


\begin{example}\label{Exmp:1PNpatche}\rm
Consider four points
\begin{equation}
\f p_{00} = (-4,3,0)^{\top}, \quad \f p_{10} = (4,-3,0)^{\top}, \quad \f p_{11} = (4,9,-8)^{\top}, \quad \f p_{01} = (-4,15,-8)^{\top},
\end{equation}
and the associated unit normal vectors
\begin{equation}
\f N_{00} = \left(-\frac{3}{5},0,-\frac{4}{5}\right)^{\top}, \,\f N_{10} = \left(\frac{3}{5},0,-\frac{4}{5}\right)^{\top}, \, \f N_{11} = \left(-\frac{2}{7},-\frac{6}{7},-\frac{3}{7}\right)^{\top}, \, \f N_{01} = \left(-\frac{6}{7},-\frac{3}{7},-\frac{2}{7}\right)^{\top}.
\end{equation}
Our goal is to construct a quadrilateral PN patch of a low degree interpolating the prescribed points and normals.

\medskip
The distribution of $\f N_{ij}$ on $\mathcal{S}^2$ shows that it is possible to use the standard stereographic projection
\begin{equation}\label{sterproj}
\pi: \mathcal{S}^2 \setminus \{ \f w \} \rightarrow \R^2,\quad (x_1,x_2,x_3)\mapsto \frac{(x_1,x_2)}{1-x_3}.
\end{equation}
with the center $\f w = (0,0,1)$. We project $\f N_{ij}$  via $\pi$ and construct the quadratic planar patch \eqref{BL_patch} in the form
\begin{equation}
\widehat{\f N}(u,v) = \frac{1}{15} \left( -3 u v+10 u-5 v-5,-4 u v-5 v \right).
\end{equation}

Then lifting $\widehat{\f N}$ back on $\mathcal{S}^2$ gives a rational spherical patch interpolating $\f N_{ij}$. Omitting the denominator we arrive at the polynomial vector field
\begin{equation}\label{Eq:stereographic2}
\f n(u,v) = (2 \widehat{\f N}, \widehat{\f N} \cdot \widehat{\f N} - 1)
\end{equation}
which fulfills the Pythagorean condition
\begin{equation}
\f n(u,v) \cdot \f n(u,v) = \left( \widehat{\f N} \cdot \widehat{\f N} + 1 \right)^2.
\end{equation}
and interpolates data $\f n(i,j) = \lambda_{ij} \, \f N_{ij}$, i.e., the prescribed normal directions. In this case we obtain the polynomial vector field of degree~4
\begin{multline}
\f n(u,v) = \frac{1}{45} \left(-18 u v+60 u-30 v-30, \,  -24 u v-30 v, \right.\\
\left. 5 u^2 v^2-12 u^2 v+20 u^2+14 u v^2-14 u v-20 u+10 v^2+10 v-40 \right).
\end{multline}
with the norm satisfying
\begin{equation}
\Vert \f n(u,v) \Vert^2 = \left[ \frac{1}{45} \left( 5 u^2 v^2-12 u^2 v+20 u^2+14 u v^2-14 u v-20 u+10 v^2+10 v+50 \right) \right]^2.
\end{equation}
Finally, we prescribe a  polynomial parameterization \eqref{Eq:PNinterpolace_surface} of degree 8 and solve the gained systems of linear equations  \eqref{Eq:PNinterpolace_rov1} and \eqref{Eq:PNinterpolace_rov2} -- in particular, we obtain  $2$-parametric solution. One particular patch from this two-parametric family of polynomial PN surfaces interpolating given data is shown in Fig~\ref{Fig:1patch} (left).
\end{example}


\begin{example}\label{Exmp:1PNtriangle}\rm
Consider three points
\begin{equation}
\f p_{00} = (0,0,0)^{\top}, \quad \f p_{10} = (10, -2, 5)^{\top},  \quad \f p_{01} = (4, 8, -3)^{\top},
\end{equation}
and the associated unit normal vectors
\begin{equation}
\f N_{00} = \left(0,0,-1\right)^{\top}, \quad \f N_{10} = \left(\frac{2}{3},-\frac{1}{3},-\frac{2}{3}\right)^{\top}, \quad \f N_{01} = \left( -\frac{2}{11},-\frac{6}{11},-\frac{9}{11} \right)^{\top}.
\end{equation}
Our goal is to construct a triangular PN patch of a low degree interpolating the prescribed points and normals.

\medskip
We use again the standard stereographic projection, cf.~\eqref{sterproj}, and construct the linear triangular planar patch \eqref{L_patch}, i.e.,
\begin{equation}
\widehat{\f N}(u,v) =  \left(\frac{1}{10}(4 u-v),\frac{1}{10} (-2 u-3 v)\right)^{\top}.
\end{equation}
Then lifting $\widehat{\f N}$ back on $\mathcal{S}^2$ and omitting the denominator yields the polynomial vector field
\begin{equation}\label{Eq:stereographic2b}
\f n(u,v) = \left(\frac{1}{5} (4 u-v),\frac{1}{5} (-2 u-3 v),\frac{1}{50} \left(10 u^2+2 u v+5 v^2-50\right)\right)^{\top}
\end{equation}
fulfilling the Pythagorean condition
\begin{equation}
\f n(u,v) \cdot \f n(u,v) = \left[ \frac{1}{50} \left(10 u^2+2 u v+5 v^2+50\right) \right]^2.
\end{equation}
Finally, we prescribe a  polynomial parameterization \eqref{Eq:PNinterpolace_surface} of degree four and solve the systems of linear equations  \eqref{Eq:PNinterpolace_rov1} and \eqref{Eq:PNinterpolace_rov2} which yields  $1$-parametric solution, see Fig.~\ref{Fig:1patch} (right) for one particular solution.
\end{example}


\medskip
In what follows, we present how the designed approach can be easily modified also for computing smoothly joined quadrilateral patches. Suppose that we are given a network of arranged points $\f p_{ij}$ with the associated unit normal vectors $\f N_{i,j}$, where $i\in \{0,1,\ldots, m\}$ and $j\in \{0,1,\ldots, n \}$. Our goal is to construct a set of $m \times n$ polynomial PN patches $\f x_{i,j}(u,v)$ for  $i\in \{1,\ldots, m \}$, $j\in \{1,\ldots, n \}$. Each patch will be defined on the interval $[0,1]\times [0,1]$ and will interpolate the corner points $\f p_{i-1,j-1}$, $\f p_{i,j-1}$, $\f p_{i-1,j}$, $\f p_{i,j}$ together with the corresponding normals. In addition, the union of these patches $\f x=\bigcup_{i,j}\f x_{i,j}$ is required to be globally $G^1$ continuous.

Using the method described above we construct the normal vector fields $\f n_{i,j}(u,v)$, $i=0,...,m$, $j=0,...,n$ for each part separately  such that the constructed $\f n_{i,j}(u,v)$ are globally $C^0$ continuous  (or joined with higher continuity when needed). We recall that local constructions as e.g. Coons patches of suitable degree, see \citep{Fa88}, are especially useful. Then for each patch we gather equations \eqref{Eq:PNinterpolace_rov1} and \eqref{Eq:PNinterpolace_rov2} which give us the whole system of linear equations corresponding to a block-structured matrix.  Finally we have to add to this system of equations additional suitable linear equations responsible for the smooth joint of the constructed patches. In particular for two patches $\f x_{i,k}$ and $\f x_{i,k+1}$, it is enough to add the following equations ensuring the $C^0$ continuity:
\begin{equation}
\f x_{i,k}(u,1)  \equiv \f x_{i,k+1}(u,0).
\end{equation}
As a result, the patches $\f x_{ik}$ and $\f x_{ik+1}$ will join with $G^1$ continuity since they have already  prescribed normal vector fields with are $C^0$ continuous.

\begin{figure}[tb]
\begin{center}
\psfrag{n00}{$\pi(\f N_{00})$}
\psfrag{n10}{$\pi(\f N_{10})$}
\psfrag{n20}{$\pi(\f N_{20})$}
\psfrag{n30}{$\pi(\f N_{30})$}

\psfrag{n01}{$\pi(\f N_{01})$}
\psfrag{n55}{$\pi(\f N_{11})$}
\psfrag{n21}{$\pi(\f N_{21})$}
\psfrag{n31}{$\pi(\f N_{31})$}

\psfrag{n02}{$\pi(\f N_{02})$}
\psfrag{n12}{$\pi(\f N_{12})$}
\psfrag{n22}{$\pi(\f N_{22})$}
\psfrag{n32}{$\pi(\f N_{32})$}

\psfrag{n03}{$\pi(\f N_{03})$}
\psfrag{n13}{$\pi(\f N_{13})$}
\psfrag{n23}{$\pi(\f N_{23})$}
\psfrag{n33}{$\pi(\f N_{33})$}
\includegraphics[width=0.35\textwidth]{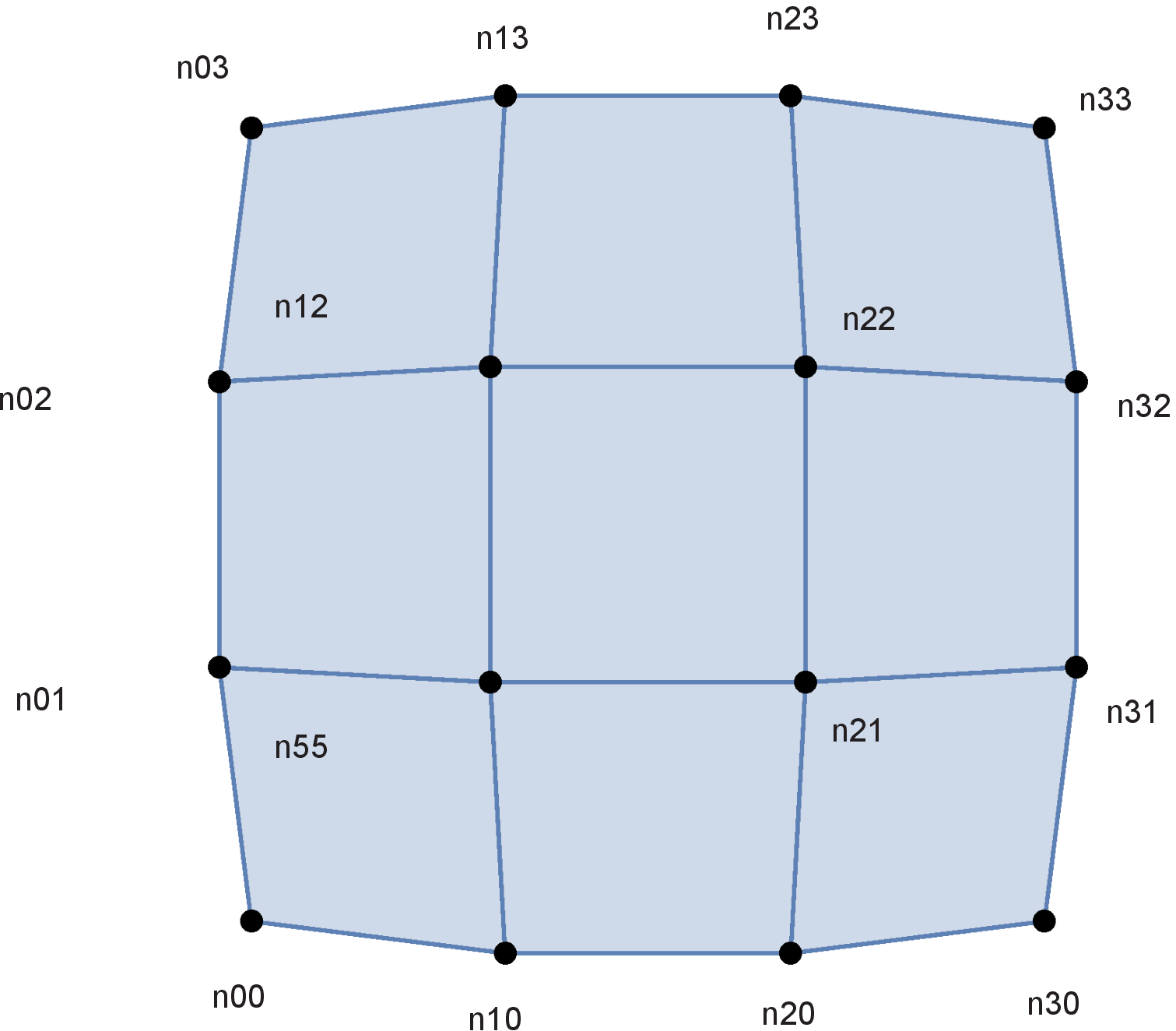}
\hspace{5ex}
\includegraphics[width=0.45\textwidth]{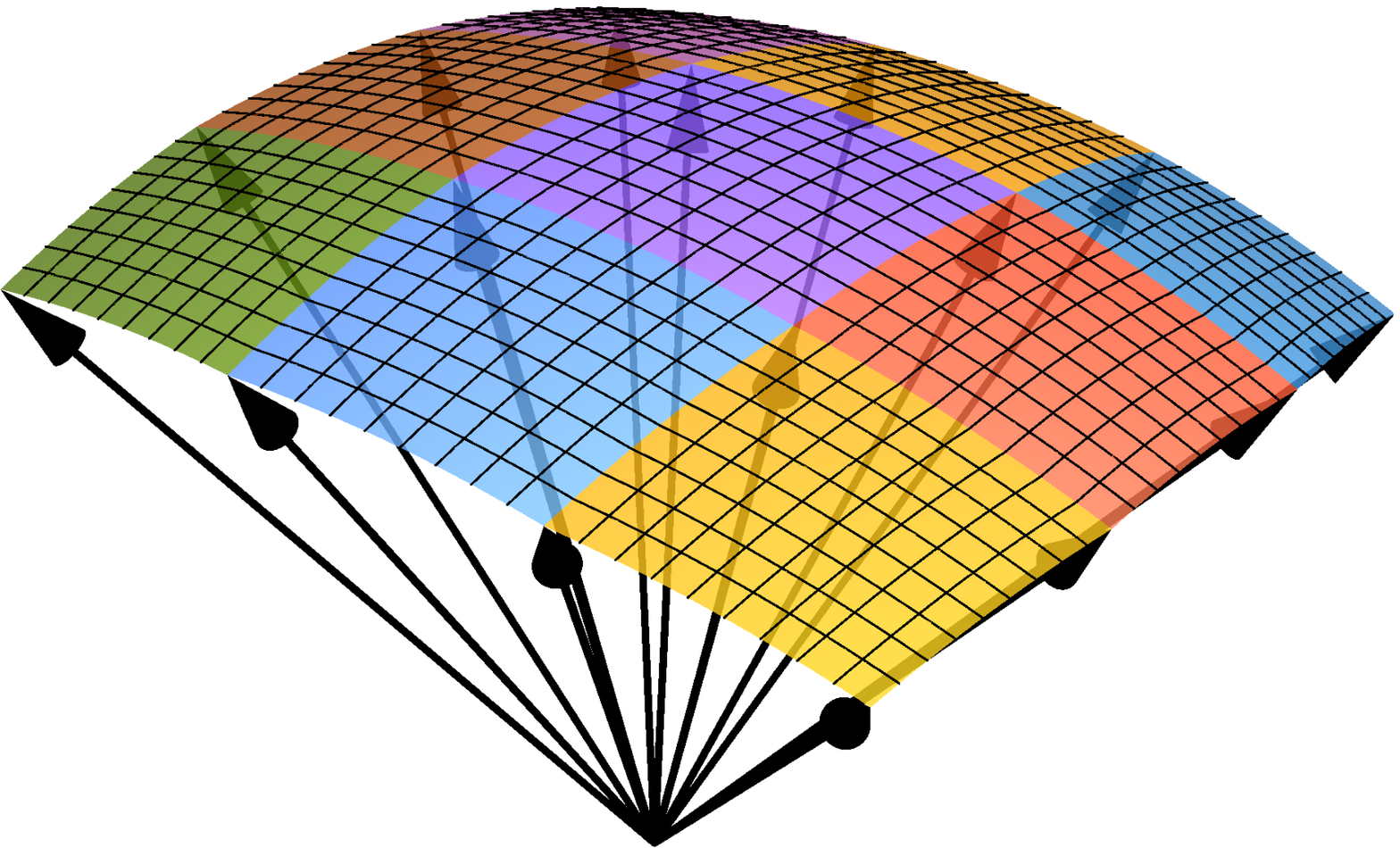}
\begin{minipage}{0.9\textwidth}
\caption{Polynomial patches $\widehat{\f N}_{ij}$ interpolating $\pi(\f N_{ij})$ in plane (left), and polynomial normal vector fields $\f n_{ij}$ interpolating $\lambda_{ij} \f N_{ij}$ (right) from Example~\ref{Exmp:9PNpatches}.}\label{Fig:9patch-fields}
\end{minipage}
\end{center}
\end{figure}

\begin{figure}[tb]
\begin{center}
\psfrag{x}{$\f{x}(u,v)$}
\psfrag{p00}{$\f p_{00}$}
\psfrag{p10}{$\f p_{10}$}
\psfrag{p22}{$\f p_{11}$}
\psfrag{p01}{$\f p_{01}$}
\psfrag{n00}{$\f n_{00}$}
\psfrag{n10}{$\f n_{10}$}
\psfrag{n22}{$\f n_{11}$}
\psfrag{n01}{$\f n_{01}$}
\includegraphics[width=0.65\textwidth]{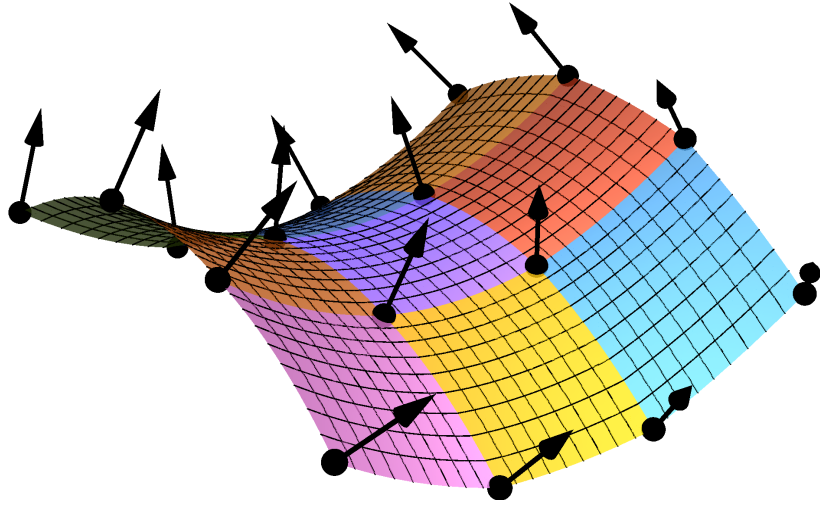}
\begin{minipage}{0.9\textwidth}
\caption{Nine polynomial PN patches $\f x_{i,j}(u,v)$ of degree $9$ interpolating points $\f p_{ij}$ and possessing tangent planes given by the unit normal vectors $\f N_{ij}$ at $\f p_{ij}$ from Example~\ref{Exmp:9PNpatches}.}\label{Fig:9patch}
\end{minipage}
\end{center}
\end{figure}

\begin{example}\label{Exmp:9PNpatches}\rm
Consider 16 points $\f p_{ij}$, $i,j=0,\ldots,3$, and the associated unit normal vectors $\f N_{ij}$, see Fig.~\ref{Fig:9patch}. The distribution of $\f N_{ij}$ on $\mathcal{S}^2$ shows again that also in this example it is possible to use the standard stereographic projection \eqref{sterproj}

We project the unit vectors $\f N_{ij}$ to plane, construct nine $C^0$ planar patches, see Fig.~\ref{Fig:9patch-fields} (left), and lift them back to space, see Fig.~\ref{Fig:9patch-fields} (right). Then, we construct nine polynomial patches of degree nine such that each patch corresponds to equations \eqref{Eq:PNinterpolace_rov1} and \eqref{Eq:PNinterpolace_rov2}. Moreover we will consider equations:
\begin{equation}
\begin{array}{lclll}
\f x_{ik}(u,1) & \equiv  &\f x_{ik+1}(u,0), & i=1,2,3 & k=1,2;\\
\f x_{kj}(1,v) & \equiv  &\f x_{k+1,j}(0,v), & j=1,2,3 & k=1,2.
\end{array}
\end{equation}
Finally, by solving the whole system of linear equations we arrive at one-parametric solution. One particular solution is shown in Fig.~\ref{Fig:9patch}.
\end{example}


\medskip
Now we present how the designed approach can be easily adapted also for constructing smoothly joined triangular patches. The following example presents computing approximated polynomial PN parameterizations of patches on given surfaces, and thus also computing approximate (piecewise) polynomial PN parameterizations either of non-PN surfaces, or of PN surfaces with rational PN parameterizations only.

\begin{example}\label{Exmp:elipsoid}\rm
Consider the ellipsoid $\mathcal{E}$ with the implicit equation
\begin{equation}\label{Eq:ellipsoid}
f(x,y,z) = 4 x^2+9y^2+9z^2-9 = 0.
\end{equation}
We approximate the ellipsoid with piecewise polynomial PN parametrization. In particular, we parameterize one octant corresponding to unit normal vectors:
\begin{equation}
\f N_{00} = \left(0,0,-1\right)^{\top}, \quad \f N_{01} = \left(0,1,0\right)^{\top}, \quad \f N_{10} = \left(1,0,0\right)^{\top}
\end{equation}
and by symmetry, we find the remainder seven octants. Solving
\begin{equation}
f(x,y,z) = 0, \quad \nabla f(x,y,z) = \alpha_{ij} \, \f N_{ij}, \qquad \alpha
_{ij} \in \R,
\end{equation}
yields two points for each normal vector. From each pair we choose one point such that all chosen points lie in the same octant, e.g.,
\begin{equation}
\f p_{00} = \left(0,0,-1\right)^{\top}, \quad \f p_{01} = \left(0,1,0\right)^{\top}, \quad \f p_{10} = \left(\frac{3}{2},0,0\right)^{\top}
\end{equation}
Next we interpolate vectors $\f n_{ij} = \lambda_{ij} \f N_{ij}$ by a polynomial vector field $\f n(u,v)$ fulfilling the Pythagorean property. In particular using stereographic projection \eqref{sterproj} we project $\f N_{ij}$ to the plane and in the plane we construct
\begin{equation}
\widehat{\f N}(u,v) = \left(\frac{\sqrt{2} u v-u v+v}{\sqrt{2} u v-2 u v+1},\frac{\sqrt{2} u v-u v+u}{\sqrt{2} u v-2 u v+1}\right)^{\top}, \quad u,v\geq0, \, u+v \leq1,
\end{equation}
as a rational triangular B\'{e}zier patch, see Fig.~\ref{Fig:el-fields}~(left).

Lifting $\widehat{\f N}(u,v)$ via $\pi^{-1}$  and omitting the denominator yields a Pythagorean normal vector field $\f n(u,v)$ of bi-degree two. The vector field was constructed such that the symmetry yields a $C^0$ continuous normal vector field of the whole ellipsoid, see Fig.~\ref{Fig:el-fields}~(right).

\begin{figure}[tb]
\begin{center}
\psfrag{N00}{$\pi(\f N_{00})$}
\psfrag{N10}{$\pi(\f N_{10})$}
\psfrag{N01}{$\pi(\f N_{01})$}
\includegraphics[width=0.25\textwidth]{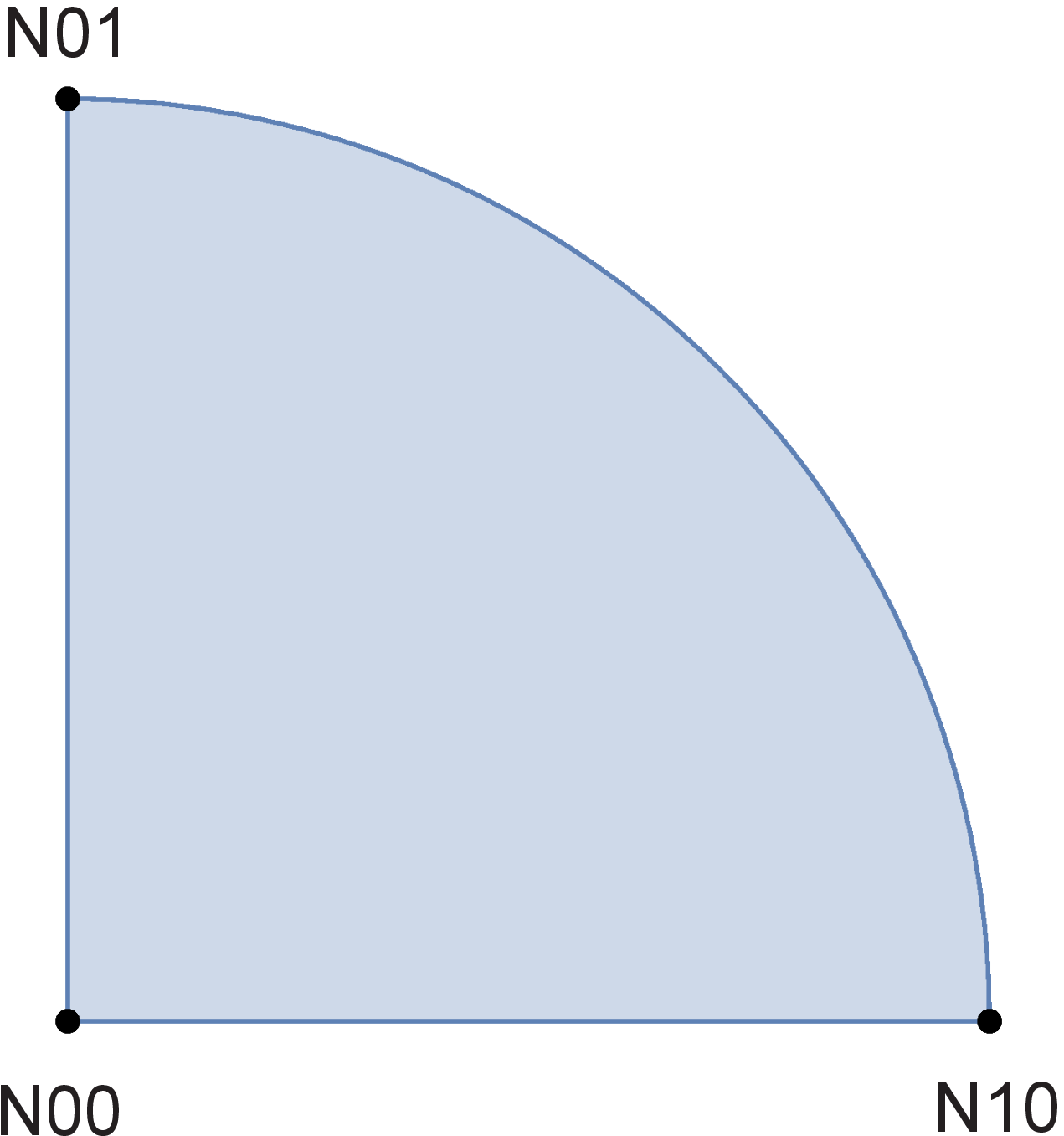}
\hspace{8ex}
\includegraphics[width=0.45\textwidth]{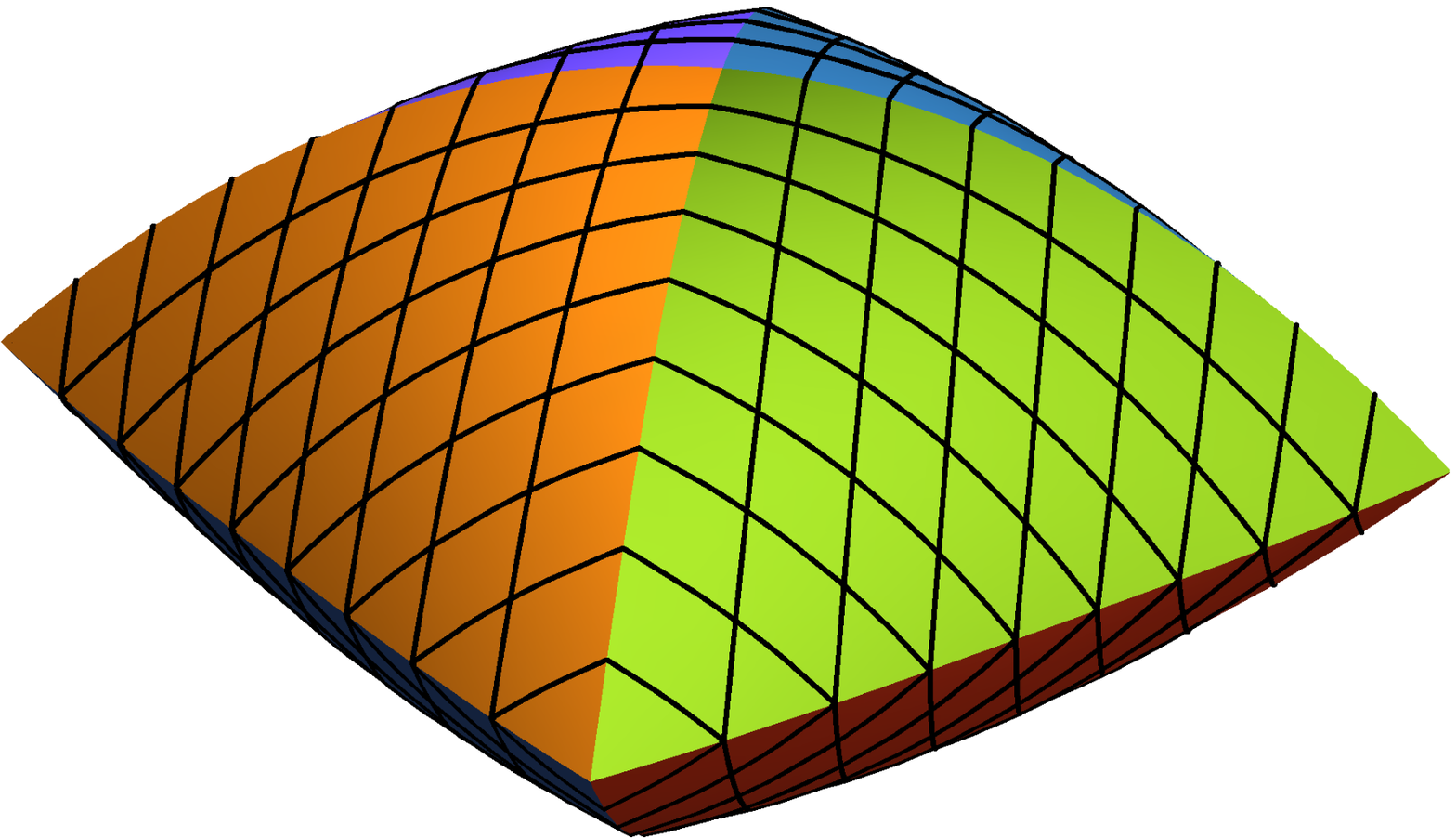}
\begin{minipage}{0.9\textwidth}
\caption{Polynomial patch $\widehat{\f N}_{ij}$ interpolating $\pi(\f N_{ij})$ in plane (left), and a corresponding polynomial PN approximation of the normal vector field of the whole ellipsoid (right) from Example~\ref{Exmp:elipsoid}.}\label{Fig:el-fields}
\end{minipage}
\end{center}
\end{figure}

Now we construct a PN surface (one triangular PN patch) of degree $12$. Solving equations \eqref{Eq:PNinterpolace_rov1} and \eqref{Eq:PNinterpolace_rov2} together with equations
\begin{equation}
\f x(u,0) \cdot (1,0,0)^{\top} \equiv 0, \qquad \f x(0,v) \cdot (0,1,0)^{\top} \equiv 0,  \qquad \f x(u,1-u) \cdot (0,0,1)^{\top} \equiv 0,
\end{equation}
which guarantee a possibility to use the symmetry and thus to obtain all remaining seven patches, yields a $5$-parametric solution. We choose the most suitable one by minimizing the following objective function
\begin{equation}
\Phi(\f t) = \int_0^1 \int_0^{1-v} \frac{f^2(\f x(u,v,\f t))}{\Vert \nabla f(\f x(u,v,\f t)) \Vert^2} \, \mathrm{d}u \, \mathrm{d}v, \quad \f t = (t_1,t_2,t_3,t_4,t_5)^{\top},
\end{equation}
which is responsible for the deviation of the parametrization $\f x(u,v,\f t)$ from the implicit surface $f=0$. In this case we obtain the error smaller then $2.4 \cdot 10^{-6}$. Finally using the symmetry we obtain the approximate piecewise polynomial PN parametric description of the whole ellipsoid \eqref{Eq:ellipsoid}, see Fig.~\ref{Fig:elipsoid}.
\end{example}

\begin{figure}[tb]
\begin{center}
\includegraphics[width=0.4\textwidth]{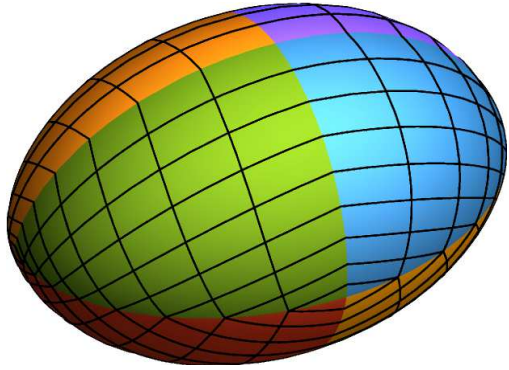}
\begin{minipage}{0.9\textwidth}
\caption{The polynomial approximate PN parametrization of the ellipsoid from Example~\ref{Exmp:elipsoid}.}\label{Fig:elipsoid}
\end{minipage}
\end{center}
\end{figure}

\begin{remark}\rm
Let emphasize that when a higher continuity of the constructed interpolation piecewise polynomial surface is needed, then the presented method can be still applied. It is enough to  increase the degree of the PN parameterizations (to have more free parameters) and add suitable extra continuity constrains (again linear equations) to the original linear system. Especially, when e.g. the $G^2$ continuity of the joint between two patches is required it is necessary to construct $C^1$ continuous normal vector fields  (e.g. applying the bi-cubic Coons construction in the quadrilateral case, or cubic Clough-Tocher or quadratic Powell-Sabin elements in the triangular case).
\end{remark}

\subsection{Hermite interpolation by piecewise polynomial medial surface transforms yielding rational envelopes}\label{Se:MOS_interpolation}

The ideas formulated in the previous section for PN surfaces can be easily adapted also for Hermite interpolation with polynomial MOS surfaces. We present the approach at least for one quadrilateral and one triangular patch. We recall that interpolations by MOS surfaces can be used, for instance, when rational blending or skinning surfaces are constructed as the envelopes of two-parameter families of spheres.

Consider four points ${\f p}_{ij} \in \R^{3,1}$, $i,j=0,1$, and four associated tangent planes ${\tau}_{ij}$ determined by the vectors ${\f t}_{ij,1}$ and ${\f t}_{ij,2}$. We find the ideal lines of ${\tau}_{ij}$, compute the conjugated lines with respect to $\Sigma$ (i.e., the ideal lines of the normal planes $\nu_{ij}$ at $\f {\f p}_{ij}$), and by intersecting them with the absolute quadric  $\Sigma$, cf. \eqref{abskvadr}, we arrive at the isotropic vectors $\f n^{\pm}_{ij}$. Next, we interpolate the isotropic Gauss image $\cal G^{\pm}$ (see Section~\ref{poly MOS}), i.e., given data $\f n^{\pm}_{ij}$, by suitable rational patches and taking them as the input for \eqref{eq MOS soustava} we arrive at an MOS patch interpolating given Hermite data $\{\f p_{ij},\tau_{ij}\}$.

\begin{example}\label{Exmp:MOS-interpolation}\rm
Consider in $\R^{3,1}$ the points
\begin{equation}
\f p_{00} = (0, 0, -3, 1)^{\top}, \quad \f p_{10} = (10, 0, 0, 2)^{\top}, \quad \f p_{11} = (10, 8, 3, 3)^{\top}, \quad \f p_{01} = (0, 8, 0, 2)^{\top},
\end{equation}
and the tangent vectors
\begin{equation}
\begin{array}{llll}
\f t_{00,1} = (1, -1, 0, 0)^{\top}, & \f t_{10,1} = (7, -7, 4, 1)^{\top}, & \f t_{11,1} = (53, -31, 1, -1)^{\top}, & \f t_{01,1} = (9, -9, -7, -3)^{\top}; \\[1ex]
\f t_{00,2} = (1, 1, 0, 0)^{\top}, & \f t_{10,2} = (7, 7, 4, 1)^{\top}, & \f t_{11,2} = (-23, 15, 1, 1)^{\top}, & \f t_{01,2} = (9, 9, 7, 3)^{\top},
\end{array}
\end{equation}
determining the tangent planes $\tau_{ij}$ at $\f p_{ij}$.

Then solving
\begin{equation}\label{Eq:iso_nor_exmp}
\langle \f n_{ij}, \f t_{ij,1} \rangle = 0, \quad \langle \f n_{ij}, \f t_{ij,2} \rangle = 0, \quad \langle \f n_{ij}, \f n_{ij} \rangle = 0,
\end{equation}
yields the following isotropic normal vectors:
\begin{equation}
\begin{array}{llll}
\f n_{00}^+ = (0,0,-1,1)^{\top}, & \f n_{10}^+ = (3,0,-4,5)^{\top}, & \f n_{11}^+ = (4,7,-4,9)^{\top}, & \f n_{01}^+ = (0,4,-3,5)^{\top}; \\[1ex]
\f n_{00}^- = (0,0,1,1)^{\top}, & \f n_{10}^- = (-5,0,12,13)^{\top}, & \f n_{11}^- = (-2,-3,6,7)^{\top}, & \f n_{01}^- = (0,-5,12,13)^{\top}.
\end{array}
\end{equation}
W.l.o.g, we choose for instance $\f n_{ij}^+$ and  compute the associated normals $\f N_{ij} = (n_1,n_2,n_3)/n_4$ on the unit sphere $\mathcal{S}^2$, i.e.,
\begin{equation}
\f N_{00} = \left(0,0,-1\right)^{\top}, \, \f N_{10} = \left(\frac{3}{5},0,-\frac{4}{5}\right)^{\top}, \, \f N_{11} = \left(\frac{4}{9},\frac{7}{9},-\frac{4}{9}\right)^{\top}, \, \f N_{01} = \left(0,\frac{4}{5},-\frac{3}{5}\right)^{\top}.
\end{equation}
Next we interpolate data  $\f N_{ij}$ by a rational patch $\f N(u,v) = (N_1/N4,N_2/N4,N_3/N4)$ on the unit sphere $\mathcal{S}^2$, see Section~\ref{Se:PN_interpolation}, and finally we arrive at $\f n^+(u,v) = (N_1,N_2,N_3,N_4)$ as the isotropic normal field interpolating data $\lambda_{ij} \, \f n_{ij}^+$.

Then we prescribe a polynomial parameterization
\begin{equation}\label{Eq:MOSinterpolace_surface}
\f x(u,v)=
\left(
\sum_{i+j\leq6} x_{1ij}{u^iv^j},
\sum_{i+j\leq6} x_{2ij}{u^iv^j},
\sum_{i+j\leq6} x_{3ij}{u^iv^j},
\sum_{i+j\leq6} x_{4ij}{u^iv^j},
\right)^{\top},
\end{equation}
of degree six and solve linear system of equations \eqref{eq:88} together with the equations:
\begin{equation}\label{Eq:MOSinterpolace_rov2}
\f x(i,j) =  \f p_{ij}, \quad \langle \f x_u(i,j), \f n_{ij}^- \rangle = 0, \quad \langle \f x_v(i,j), \f n_{ij}^- \rangle = 0, \quad i,j=0,1.
\end{equation}
Let us emphasize that equations \eqref{Eq:MOSinterpolace_rov2} must be added to ensure the prescribed interpolation conditions, i.e., that $\f x(u,v)$ is tangent to $\tau_{ij}$ at $\f p_{ij}$. Finally we obtain $8$-parametric set of polynomial MOS surfaces of degree six interpolating given Hermite data $\{\f p_{ij}, \tau_{ij} \}$, see Fig.~\ref{Fig:MOS} (left) for one particular example from the set of solutions.
\end{example}

The triangular patch would be treated analogously to the quadrilateral one, see the following example.

\begin{example}\label{Exmp:MOS-interpolation2}\rm
Consider in $\R^{3,1}$ three points
\begin{equation}
\f p_{00} = (0, 0, -4, 1)^{\top}, \quad \f p_{10} = (8, -5, 0, 2)^{\top}, \quad \f p_{01} = (3, 6, 0, 2)^{\top},
\end{equation}
and the three pairs of tangent vectors
\begin{equation}
\begin{array}{lll}
\f t_{00,1} = (1, -1, 0, 0)^{\top}, & \f t_{10,1} = (8, -8, 9, 2)^{\top}, & \f t_{01,1} = (41, -41, -15, -7)^{\top}; \\[1ex]
\f t_{00,2} = (1, 1, 0, 0)^{\top}, & \f t_{10,2} = (16, 16, 5, 2)^{\top}, & \f t_{01,2} = (41, 41, 61, 23)^{\top},
\end{array}
\end{equation}
determining the tangent planes $\tau_{ij}$ at $\f p_{ij}$. By Solving \eqref{Eq:iso_nor_exmp} we arrive at the isotropic normal vectors:
\begin{equation}
\begin{array}{llll}
\f n_{00}^+ = (0,0,-1,1)^{\top}, & \f n_{10}^+ = (2,-1,-2,1)^{\top}, & \f n_{01}^+ = (4, 7, -4, 9)^{\top}; \\[1ex]
\f n_{00}^- = (0,0,1,1)^{\top}, & \f n_{10}^- = (-3,2,6,7)^{\top}, & \f n_{01}^- = (-2, -3, 6, 7)^{\top}.
\end{array}
\end{equation}
Again, we choose e.g. $\f n_{ij}^+$, compute the associated normals $\f N_{ij} = (n_1,n_2,n_3)/n_4$ on the unit sphere $\mathcal{S}^2$ and construct the spherical triangular patch interpolating $\f N_{ij}$ (see Section~\ref{Se:PN_interpolation}), i.e.,
\begin{equation}\label{Eq:MOS_troj_nor}
\f N(u,v) = \left(\frac{52 u+40 v}{13 u^2+2 u v+25 v^2+65},\frac{70 v-26u}{13 u^2+2 u v+25 v^2+65},\frac{13 u^2+2 u v+25 v^2-65}{13u^2+2 u v+25 v^2+65}\right)^{\top}.
\end{equation}
Then we arrive at $\f n^+(u,v) = (N_1,N_2,N_3,N_4)$ as the isotropic normal field interpolating data $\lambda_{ij} \, \f n_{ij}^+$, where $(N_1/N_4,N_2/N_4,N_3/N_4)$ is given by \eqref{Eq:MOS_troj_nor}.

Finally we prescribe a polynomial quartic parameterization \eqref{Eq:MOSinterpolace_surface} and solving linear system of equations \eqref{eq:88} together with the equations \eqref{Eq:MOSinterpolace_rov2} (now for $i+j<2$) yields $7$-parametric set of quartic polynomial MOS surfaces interpolating given Hermite data $\{\f p_{ij}, \tau_{ij} \}$, see Fig.~\ref{Fig:MOS} (right) for one chosen triangular patch from the set of all solutions.
\end{example}

\begin{figure}[tb]
\begin{center}
\psfrag{A}{ }
\psfrag{a}{$\f{x}(u,v)$}
\psfrag{au}{$\f{x}_u$}
\psfrag{av}{$\f{x}_v$}
\psfrag{T}{$\tau(u,v)$}
\psfrag{ilt}{}
\psfrag{n1}{$\f{n}_1$}
\psfrag{n2}{$\f{n}_2$}
\psfrag{pi}{$\pi_{\infty}:x_0=0$}
\includegraphics[width=0.47\textwidth]{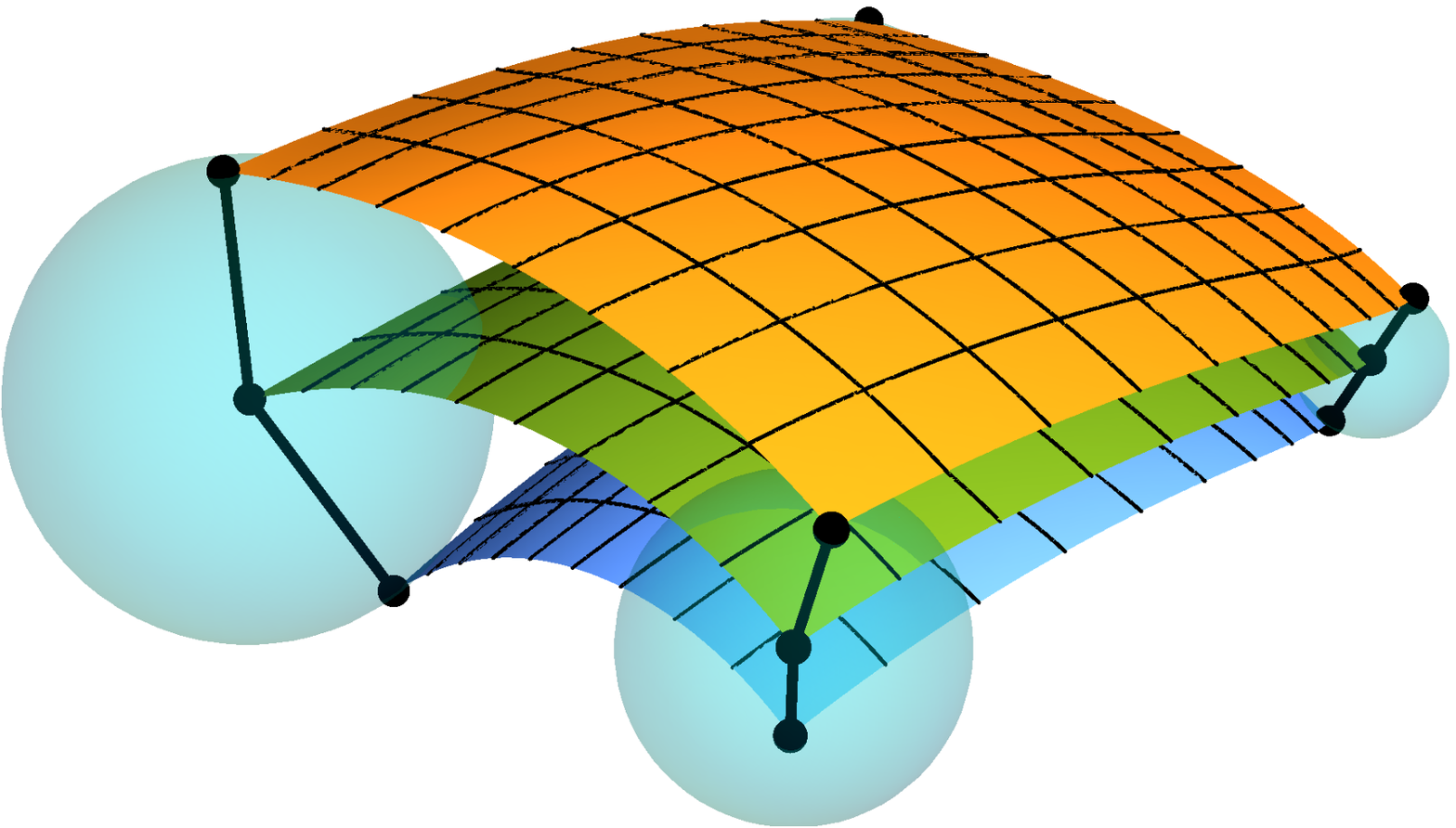}
\hspace{3ex}
\includegraphics[width=0.47\textwidth]{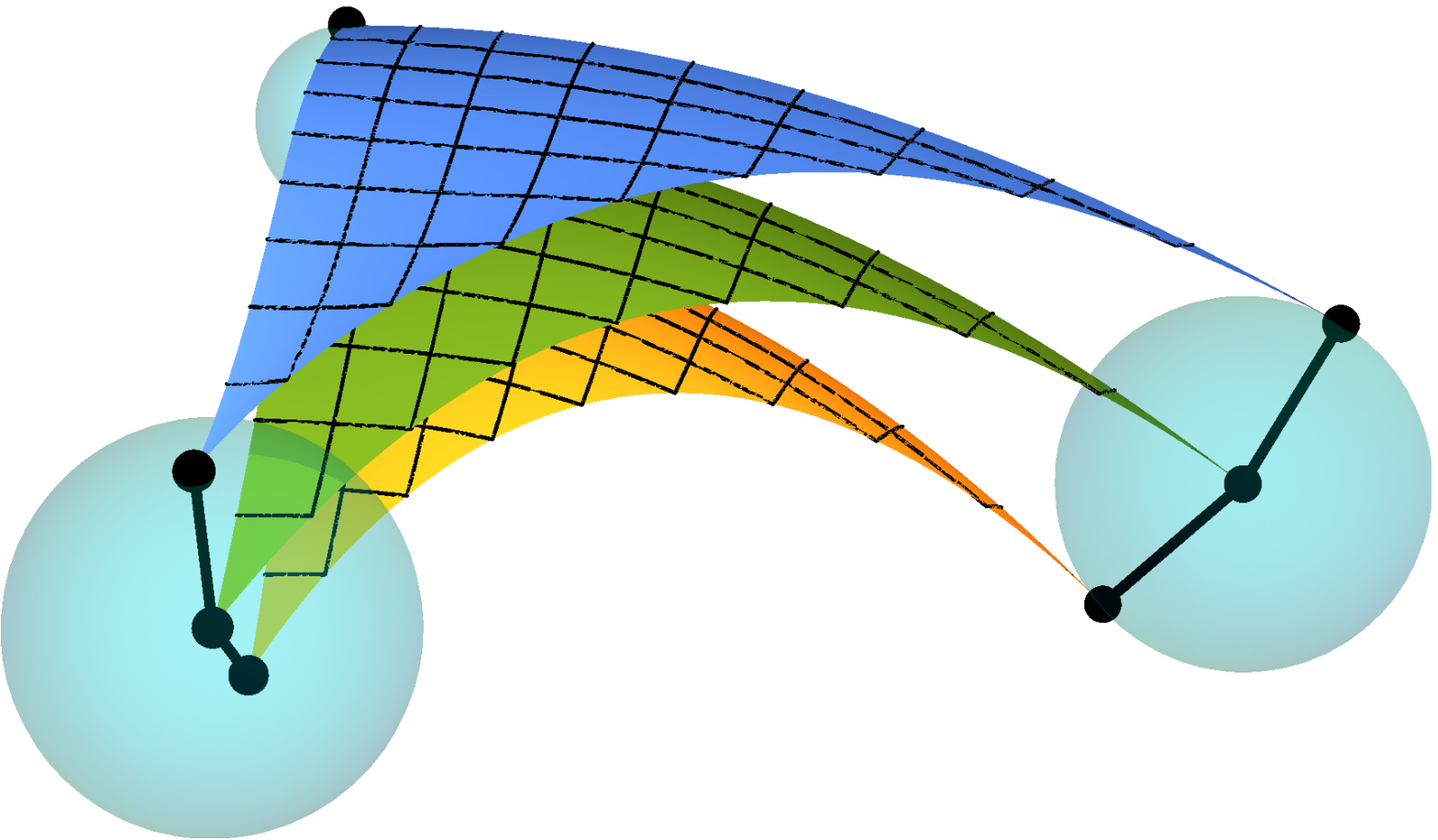}
\begin{minipage}{0.9\textwidth}
\caption{Quadrilateral and triangular medial surfaces (green) with associated rational envelopes (yellow, blue) from Examples~ \ref{Exmp:MOS-interpolation} and \ref{Exmp:MOS-interpolation2}, respectively.}\label{Fig:MOS}
\end{minipage}
\end{center}
\end{figure}

\section{Concluding remarks}\label{sec_concl}

In this paper the problem of Hermite interpolations by piecewise
polynomial surfaces with polynomial area element was investigated.
It was shown that the interpolation problem can be always transformed to solving a
system of linear equations and the same approach is suitable not only
for polynomial PN surfaces but also for polynomial MOS surfaces.
Simplicity and functionality of the designed algorithm was presented on
several examples. In our future work we would like to focus on better
understanding of the quantity~$\Delta$ in \eqref{Eq:odhad} responsible for increasing free
parameters in the construction, on the study of existence (or its
eliminating) of the factor $f(u,v)$ which causes vanishing of the normals
along a curve on the surface, and finally on the construction of
polynomial PN patches given by the boundary curves, which is a
challenging open problem in geometric modelling.

\section*{Acknowledgments}

The authors Michal Bizzarri, Miroslav L\'{a}vi\v{c}ka and Jan Vr\v{s}ek were supported by the project LO1506 of the Czech Ministry of Education, Youth and Sports.

\bigskip
\begin{appendix}
\section{Appendix}


\smallskip
Let $\f N(u,v,w)$ be the homogenization of the normal vector field $\f n(u,v)$, i.e., it is a triple of homogeneous polynomials $N_i(u,v,w)$ of the same degree $k$. As already mentioned earlier, the result depends on the occurrence of the base points of $\f N$ over $\C$.  Denote $R=\C[u,v,w]$ the coordinate ring of $\mathbb{P}^2_\C$. Analogously to $\mathrm{Syz}(\f n)$ we define $\mathrm{Syz}(\f N)$ which is a submodule of $R^3$ in this case.

Next, it is known that $R$ is a graded module over itself whose graded pieces $R_i$ are formed by the  sets of homogeneous polynomials of degree $i$. Obviously each $R_i$ is a finite--dimensional vector space over $\C$. For such a~graded module  $S$, the {\em Hilbert function} is  defined as
\begin{equation}
  \hf(S,\ell)=\dim_{\C}S_{\ell}.
\end{equation}
We introduce the standard notation $R(j)$ for the shifted module, i.e.,  $R(i)_j=R_{i+j}$. Hence the Hilbert function of this module is
\begin{equation}\label{eq hf shifted}
\hf(R(j),\ell)={{j+\ell+2}\choose{2}}.
\end{equation}
Let us emphasize that unlike $\mathrm{Syz}(\f n)$ the homogeneous syzygy module $\mathrm{Syz}(\f N)$ is not free anymore. The reason is that the basis of $\mathrm{Syz}(\f n)$ does not remain a basis of syzygy module after homogenization. To see this  let $\f N(u,v,w)=(2uw,2vw,w^2-u^2-v^2)$ be the homogenization of the vector field $\f n$ from Example~\ref{ex mu basis}. Then obviously $\f P=(v,-u,0)\in\mathrm{Syz}(\f N)$ but there is no way how to write it as $\f R[u,v,w]$-linear combination of
\begin{equation}
 \f Q=(u^2-w^2,uv,2uw) \quad\text{and}\quad\f R=(uv, v^2-w^2,2 vw).
\end{equation}
In fact $\mathrm{Syz}(\f N)$ is generated by $\f P$, $\f Q$ and $\f R$. Nevertheless they do not form a basis because of the dependence relation $ w^2 \f P + v \f Q - u \f R = 0$.

Next let $I=\langle N_1,N_2,N_3\rangle$  denote the ideal generated by the components of the normal field $\f N$. Then there exists the so called {\em Koszul complex}
\begin{equation}\label{eq koszul}
\begin{array}{c}
\xymatrix{
0\ar[r]&  R(-3k)\ar[r]^-{\delta_3}& \displaystyle \bigoplus_{i=1}^3R(-2k)\ar[r]^-{\delta_2}& \displaystyle \bigoplus_{i=1}^3R(-k)\ar[r]^-{\delta_1}&I\ar[r]&0
}
\end{array},
\end{equation}
where the differentials are given by
\begin{equation}\label{eq koszul differentials}
  \delta_3=\left[
  \begin{array}{c}
    N_3\\-N_2\\N_1
  \end{array}
  \right],
  \quad
   \delta_2=\left[
  \begin{array}{ccc}
    N_2&N_3&0\\
    -N_1&0&N_3\\
    0&-N_1&-N_2
  \end{array}
  \right]
  \quad\text{and}\quad
   \delta_1=\left[
  \begin{array}{ccc}
    N_1 &N_2& N_3
  \end{array}
  \right]
\end{equation}
This complex is known to be exact if and only if the sequence $N_i$, $i=1,2,3$, is regular. If $\f N$ admits the base points then the sequence cannot be regular. Nevertheless we have, cf. \citep{CoSch03}

\begin{lemma}\label{lem cox schenck}
  If $I=\langle N_1,N_2,N_3\rangle$ has a codimension two in $R$ then complex \eqref{eq koszul} is exact except at $\bigoplus_{i=1}^3R(-k)$.
\end{lemma}

\smallskip
Now we can formulate and prove the theorem that gives consequently a result presented in Section~\ref{poly PN}.
\begin{theorem}\label{thm formule}
Let $\f N(u,v,w)$ be a homogeneous normal vector field of degree $k$ as above. Then
\begin{equation}\label{eq sum of hilbert}
 \hf(\mathrm{Syz}(\f N),\ell)\geq 3{{\ell-k+2}\choose{2}}-{{\ell-2k+2}\choose{2}},
\end{equation}
where the equality holds if and only if the normal field is basepoint-free.
\end{theorem}

\begin{proof}
From \eqref{eq koszul} and \eqref{eq koszul differentials} we immediately see that $\mathrm{Syz}(\f N)$ is the kernel of the map $\delta_1$. Hence
\begin{equation}\label{eq hf syz ker}
  \hf(\mathrm{Syz}(\f N),\ell)=\hf(\ker\delta_1,\ell+k),
\end{equation}
where the $+k$ term comes from the shift $R(-k)\supset\ker(\delta_1)$.

By Lemma~\ref{lem cox schenck} complex \eqref{eq koszul} is exact except at $\bigoplus_{i=1}^3R(-k)$, hence it is possible to express Hilbert function of the $\mathrm{im}(\delta_2)$
\begin{equation}\label{eq hf im delta2}
  \hf(\mathrm{im}(\delta_2))= \hf(R^3(-2k))-\hf(R(-3k)),
\end{equation}
Now $\mathrm{im}(\delta_2)$ is a submodule of $\ker(\delta_1)$ and moreover $\mathrm{im}(\delta_2)=\ker(\delta_1)$ if and only if $\f N$ has no base point. Thus we have $\hf(\mathrm{ker}(\delta_1))\geq \hf(\mathrm{im}(\delta_2))$ where the equality occurs whenever the normal field is basepoint-free. Substituting \eqref{eq hf shifted} into \eqref{eq hf im delta2} proves the theorem.

%
\end{proof}

To sum up, Lemma~\ref{thm jakkoliv} is only a direct reformulation of Theorem \ref{thm formule} because the following identity holds
\begin{equation}
3{{\ell-k+2}\choose{2}}-{{\ell-2k+2}\choose{2}} = 3{{\ell+2}\choose{2}}-{{\ell+k+2}\choose{2}},
\end{equation}
and $\hf(\mathrm{Syz}(\f N),\ell)$ is exactly the dimension of the set of fields $\f q$ of degree at most $\ell$ orthogonal to $\f n$.

\end{appendix}

\bibliographystyle{elsarticle-harv}
\bibliography{bibliography}{}

%

\end{document}